\documentclass[]{article}

\usepackage[letterpaper]{geometry}

\usepackage{amsmath}
\usepackage{amssymb}
\usepackage{numprint}
\usepackage{hyperref}
\usepackage{graphicx}
\usepackage[title,toc,titletoc,page]{appendix}
\usepackage{flushend}
\usepackage{xcolor}
\usepackage{wrapfig}
\usepackage{graphicx}
\usepackage{subcaption}
\usepackage[numbers,sort&compress]{natbib}
\usepackage{nccmath}
\usepackage{todonotes}
\usepackage{fancyhdr}

\newtheorem{lemma}{Lemma}
\newtheorem{proof}{Proof}
\newtheorem{claim}{Claim}

\def\MdN{\ensuremath{\mathbb{N}}}

\newcommand{\cut}{\mathcal{C}}
\newcommand{\bestcut}{\widehat{\cut}}
\newcommand{\wgt}{\mathcal{W}}
\newcommand{\bestwgt}{\widehat{\wgt}}
\newcommand{\vopt}{\mathcal{V}}
\newcommand{\queue}{\mathcal{Q}}
\newcommand{\Oh}[1]{\mathcal{O}\!\left( #1\right)}
\newcommand{\Is}       {:=}

\newif\ifDoubleBlind
\DoubleBlindfalse

\newcommand{\ie}{i.e.\ }
\newcommand{\etal}{et~al.~}

\newcommand{\CC}{C\texttt{++}}

\newcommand*{\shorterDots}{.\kern-0.06em.\kern-0.06em.} 

\newcommand{\papertitle}{Shared-Memory Branch-and-Reduce for Multiterminal Cuts}
\ifDoubleBlind
\newcommand{\mytitle}{\papertitle}
\else
\newcommand{\mytitle}{\papertitle \thanks{
    The research leading to these results has received funding from the European Research Council under the European Community's Seventh Framework Programme (FP7/2007-2013) /ERC grant agreement No. 340506}}
\fi{}

\date{}

\ifDoubleBlind
\author{Double Blind}
\else
\author{Monika Henzinger\thanks{University of Vienna, Faculty of Computer Science, Vienna, Austria} 
\and Alexander Noe \thanks{University of Vienna, Faculty of Computer Science, Vienna, Austria}
\and Christian Schulz \thanks{University of Vienna, Faculty of Computer Science, Vienna, Austria}}
\fi{}

\begin{document}

\title{\mytitle}

\maketitle

\begin{abstract} 

We introduce the fastest known exact algorithm~for~the multiterminal cut problem with $k$ terminals.
In particular, we engineer existing as well as new data reduction rules. 
We use the rules within a branch-and-reduce framework and to boost the performance of an ILP formulation. 
Our algorithms achieve improvements in running time of up to \emph{multiple orders of magnitudes} over the ILP formulation without data reductions, which has been the de facto standard used by practitioners. This allows us to solve instances to optimality that are significantly larger than was previously possible.

\end{abstract}
\ifDoubleBlind
\vfill
\clearpage
\setcounter{page}{1}
\pagenumbering{arabic}
\fancyfoot[C]{\thepage}
\pagestyle{fancy}
\else
\fi{}

\section{Introduction}
We consider the multiterminal cut problem with $k$ terminals. 
Its input is an undirected edge-weighted graph $G=(V,E,w)$ with edge weights $w: E \mapsto \MdN_{>0}$ and its goal is to divide its set of nodes into~$k$ blocks such that each blocks contains exactly one terminal and the weight sum of the edges running between the~blocks~is~minimized. 
The problem has applications in a wide range of areas, for example in multiprocessor scheduling~\cite{DBLP:journals/tse/Stone77}, clustering~\cite{DBLP:journals/njc/PferschyRW94} and bioinformatics~\cite{karaoz2004whole,nabieva2005whole,vazquez2003global}. It is a fundamental combinatorial optimization problem which was first formulated by Dahlhaus~\etal\cite{dahlhaus1994complexity} and Cunningham~\cite{cunningham1989optimal}. It is NP-hard for $k \geq 3$~\cite{dahlhaus1994complexity}, even on planar graphs, and reduces to the minimum $s$-$t$-cut problem, which is in P, for $k=2$. The minimum $s$-$t$-cut problem aims to find the minimum cut in which the vertices $s$ and $t$ are in different blocks. Most algorithms for the minimum multiterminal cut problem use minimum s-t-cuts as a subroutine. Dahlhaus~\etal\cite{dahlhaus1994complexity} give a $2(1-1/k)$ approximation algorithm with polynomial running time. Their approximation algorithm uses the notion of \emph{isolating cuts}, \ie the minimum cut separating a terminal from all other terminals. They prove that the union of the $k-1$ smallest isolating cuts yields a valid multiterminal cut with the desired approximation ratio. The currently best known approximation algorithm by Buchbinder~\etal\cite{buchbinder2013simplex} uses linear program relaxation to achieve an approximation ratio of $1.323$.

While the multiterminal cut problem is NP-hard, it is \emph{fixed-parameter tractable} (FPT), parameterized by the multiterminal cut weight $\wgt(G)$. A problem is fixed-parameter tractable if there is a parameter $\sigma$ so that there is an algorithm with runtime $f(\sigma) \cdot n^{\Oh{1}}$. Marx~\cite{marx2006parameterized} proves that the multiterminal cut problem is FPT and Chen~\etal\cite{chen2009improved} give the first FPT algorithm with a running time of $4^{\wgt(G)}\cdot n^{\Oh{1}}$, later improved by Xiao~\cite{xiao2010simple} to $2^{\wgt(G)}\cdot n^{\Oh{1}}$ and by Cao~\etal\cite{cao20141} to $1.84^{\wgt(G)}\cdot n^{\Oh{1}}$. However, to the best of our knowledge, there is no actual implementation for any~of~these~algorithms. 

The minimum $s$-$t$-cut problem and its equivalent counterpart, the maximum $s$-$t$-flow problem~\cite{ford2015flows} were first formulated by Harris~\etal\cite{harris1955fundamentals}. Ford and Fulkerson~\cite{ford1956maximal} gave the first algorithm for the problem with a running time of $\Oh{mn\wgt}$. One of the fastest known algorithms in practice is the push-relabel algorithm of Goldberg and Tarjan~\cite{goldberg1988new} with a running~time~of~$\Oh{mn\log(n^2/m)}$.

Problems related to the minimum multiterminal cut problem also appear in the data mining community, namely the very similar and heavily studied \emph{seed expansion problem}, for which the aim is to find ground-truth clusters when given a small subset of the cluster vertices. In contrast to the minimum multiterminal cut problem, these clusters might overlap. There is a multitude of approaches adding and removing vertices greedily~\cite{andersen2006communities,clauset2005finding,luo2008exploring,mislove2010you}. PageRank~\cite{page1999pagerank} is reported to be well suited for the problem~\cite{kloumann2014community} and there are multiple approaches that aim to make PageRank perform even better~\cite{andersen2006local,bian2017many,leskovec2010empirical}. Another approach is to use machine learning methods such as geometric~\cite{ye2017learning} or relational~\cite{macskassy2003simple} neighborhood classifiers.

Closely related to the problem is also the minimum cut problem. For this problem, the goal is to divide the set of nodes in an undirected edge-weighted graph into two blocks while minimizing the weight sum of the cut edges. Both Padberg~\etal\cite{padberg1990efficient} and Nagamochi~\etal\cite{nagamochi1992computing,nagamochi1994implementing} give local conditions that are sufficient to contract edges such that the global minimum cut is maintained (and hence the problem size is reduced). An efficient implementation of those conditions is given by Henzinger~\etal\cite{henzinger2017local}. In this work, we adapt the conditions from their works that are applicable to the minimum multiterminal cut problem and use them to reduce the size of the problem.

Our paper has the following \emph{main contributions}: 
We engineer existing as well as new data reduction rules for the minimum multiterminal cut problem with $k$ terminals. 
These reductions are used within a branch-and-reduce framework as well as to boost the performance of an ILP formulation for the problem.
Through extensive experiments we show that kernelization has a significant impact on both, the branch-and-reduce framework as well as the ILP formulation.
Our experiments also show a clear trade-off: combining reduction rules with the ILP is very fast for problems which have a small kernel but a high cut value and the fixed-parameter tractable branch-and-reduce algorithm is highly efficient when the cut value is small.
Overall, we obtain algorithms that are multiple orders of magnitude faster than the ILP formulation which is de facto standard to solve the problem to optimality.
\section{Preliminaries}\label{s:preliminaries}

\subsection{Basic Concepts}
Let $G = (V, E, w)$ be a weighted undirected graph with vertex set $V$, edge set $E \subset V \times V$ and
non-negative edge weights $w: E \rightarrow \mathbb{N}$. 
We extend $w$ to a set of edges $E' \subseteq E$ by summing the weights of the edges; that is, $w(E')\Is \sum_{e=(u,v)\in E'}w(u,v)$. 
Let $n = |V|$ be the
number of vertices and $m = |E|$ be the number of edges in $G$. The \emph{neighborhood}
$N(v)$ of a vertex $v$ is the set of vertices adjacent to $v$. The \emph{weighted degree} of a vertex is the sum of the weights of its incident edges.
For a set of vertices $A\subseteq V$, we denote by $E[A]\Is \{(u,v)\in E \mid u\in A, v\in V\setminus A\}$; that is, the set of edges in $E$ that start in $A$ and end in its complement.
A \emph{$k$-cut}, or \emph{multicut}, is a partitioning of $V$ into $k$ disjoint non-empty blocks, \ie $V_1 \cup \dots \cup V_k = V$. The weight of a $k$-cut is defined as the weight sum of all edges crossing block boundaries, \ie $w(E \cap \bigcup_{i<j}V_i\times V_j)$.

\subsection{Multiterminal Cuts}

A \emph{multiterminal cut} for $k$ terminals $T = \{t_1,\shorterDots,t_k\}$ is a multicut with~$t_1 \in V_1,\shorterDots,t_k \in V_k$. Thus, a multiterminal cut pairwisely separates all terminals from each other. The edge set of the multiterminal cut with minimum weight of $G$ is called $\cut(G)$ and the associated optimal partitioning of vertices is denoted as $\vopt = \{\vopt_1, \dots, \vopt_k\}$.
 $\cut$ can be seen as the set of all edges that cross block boundaries in $\vopt$, \ie $\cut(G) = \bigcup \{e = (u,v) \mid \vopt_u \neq \vopt_v\}$. The weight of the minimum multiterminal cut is denoted as $\wgt(G) = w(\cut(G))$. At any point in time, the best currently known upper bound for $\wgt(G)$ is denoted as $\bestwgt(G)$ and the best currently known multiterminal cut is denoted as $\bestcut(G)$. If graph $G$ is clear from the context, we omit it in the notation. 
 There may be multiple minimum multiterminal cuts, however, we aim to find one multiterminal cut with minimum weight. 

In this paper we use \emph{minimum s-T-cuts}. For a vertex $s$ (\emph{source}) and a non-empty vertex set $T$ (\emph{sinks}), the minimum s-T-cut is the smallest cut in which $s$ is one side of the cut and all vertices in $T$ (except for $s$, if $s\in T$) are on the other side. This is a generalization of minimum s-t-cuts that allows multiple vertices in $t$ and can be easily replaced by a minimum s-t-cut by connecting every vertex in $T$ with a new super-sink by infinite-capacity edges. We denote the capacity of a minimum-s-T-cut, \ie the sum of weights in the smallest cut separating $s$ from $T$, by $\lambda(G,s,T)$.

The example in Figure~\ref{fig:example_mtcut} shows a graph with $4$ terminals. The minimum s-T-cut for each terminal with $T$ being the set of all terminals is shown in red and the minimum multiterminal cut is shown in blue. We can see that any $k-1$ minimum s-T-cuts (in red) separate all terminals and are thus a valid multiterminal cut. 
In our algorithm we use \emph{graph contraction} and \emph{edge deletions}.
Given
an edge $e = (u, v) \in E$, we define $G/e$ to be the graph  
after \emph{contracting} $e$.
In the contracted graph, we delete vertex $v$ and all incident edges. For each edge $(v, x) \in E$, we add an edge $(u, x)$
with $w(u, x) = w(v, x)$ to $G$ or, if the edge already exists, we give it the edge
weight $w(u,x) + w(v,x)$. For the \emph{edge deletion} of an edge $e$, we define $G-e$ as the graph $G$ in which $e$ has been removed. Other vertices and edges remain the same.

For a given multiterminal cut $S$, the graph $G\backslash S$ splits $G$ into $k$ blocks as defined by the cut edges in $S$, each containing exactly one terminal. Let the residual $R(t_i)$ be the connected component of $G \backslash S$ containing $t_i$ and $\delta(t_i) = |E(R(t_i), V \backslash R(t_i))|$ be the edges in $S$ incident to $t_i$.

\begin{figure}[t]
   \centering
   \includegraphics[width=.45\textwidth]{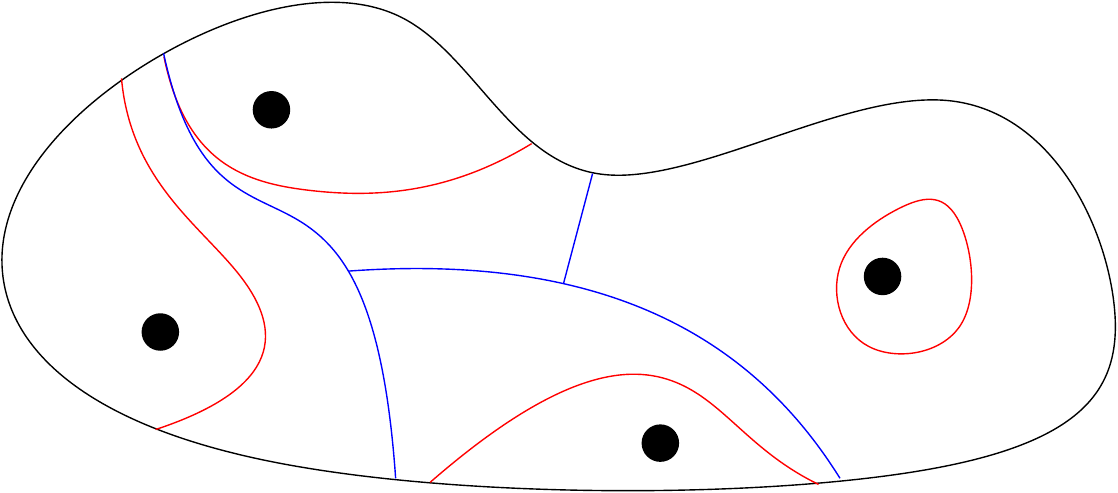}
   \caption{\label{fig:example_mtcut}Graph with $4$ terminals. Minimum $s$-$T$-cut for each terminal shown in red, $\cut$ in blue}
 \end{figure}

\section{Branch and Reduce for Multiterminal Cut}
\label{s:algorithm}

In this section we give an overview of our approach to find the optimal multiterminal cut in large graphs. Our algorithm combines kernelization techniques with an engineered bounded search. 

We begin by finding all connected components of~$G$. We can then look at all connected components independently from each other, as there is a trivial cut of weight $0$ between different connected components. If a connected component contains only one terminal $t$, it can be separated from all other terminals by using the whole connected component as the block $\vopt_t$ belonging to terminal $t$. Due to it being not connected to any other terminals, the cut value is $0$. If a connected component contains no terminals, the result $\wgt$ is identical no matter which block $\vopt$ the connected component belongs to. For a connected component $C$ with two terminals $s$ and $t$, we can run a minimum s-t-cut algorithm on $C$ to find the minimum cut. The optimal blocks $\vopt_s$ and $\vopt_t$ then consist of the two sides of the s-t-cut. On a connected component with more than two terminals, the problem is NP-hard~\cite{dahlhaus1994complexity}. We run our branch and reduce algorithm on this component. As those runs are completely independent, we only look at one connected component in the following and disregard the rest of the graph for now. 

\begin{figure}[t!]
  \centering
  \includegraphics[width=\textwidth]{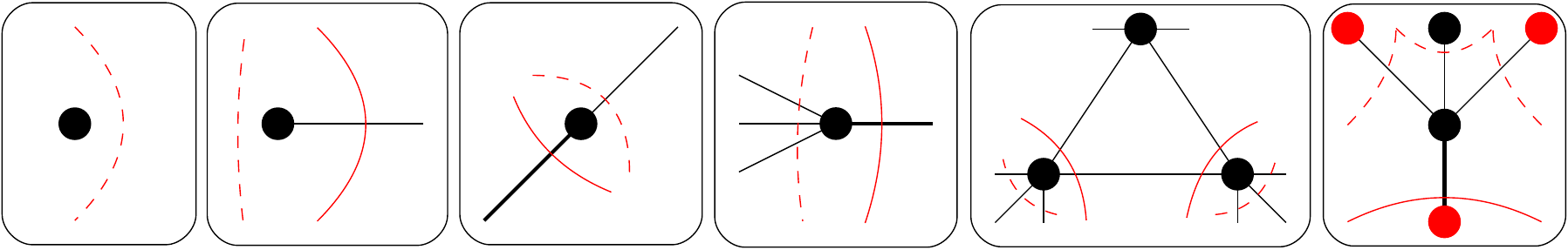}
  \caption{\label{fig:various_reductions} Reductions. Solid line cannot be minimal as dashed line has smaller weight: (1) \texttt{IsolatedVertex}, (2) \texttt{DegreeOne}, (3) \texttt{DegreeTwo} , (4) \texttt{HeavyEdge}, (5)~\texttt{HeavyTriangle} and (6) \texttt{SemiEnclosedVertex}}
\end{figure}

For a graph $G$, we can find an upper bound $\bestwgt$ which is equal to the sum of minimum s-T-cut weights minus the heaviest of them. $\bestwgt(G) = \sum_{s \in T} \lambda(G,s,T\backslash\{s\}) - \arg\,max_{s \in T} \lambda(G,s,T\backslash\{s\})$. However, this is not necessarily the minimum multiterminal cut. 

We can also give a lower bound for the minimum multiterminal cut: as $\lambda(G,s,T\backslash\{s\})$ is by definition minimal, $\cut$ has at least as many edges incident to terminal $s$ as $\lambda(G,s,T\backslash\{s\})$. As this is true for every terminal (and every edge is only incident to two vertices), $\cut(G) \cdot 2 \geq \sum_{s \in T} \lambda(G,s,T\backslash\{s\})$, so that $\cut(G) \geq \sum_{s \in T} \lambda(G,s,T\backslash\{s\}) / 2$.  

In our algorithm, we keep a queue $\queue$ of problems. A problem in $\queue$ consists of a graph $G_{\queue}$, a set of terminals, the upper and lower bound for $\wgt(G_{\queue})$ and the weight sum of all deleted edges in $G_{\queue}$. When our algorithm is initialized, $\queue$ is initialized with a single problem, whose graph is $G$ and whose set of terminals is $T$. The problem has $0$ deleted edges and its lower and upper bound for $\wgt(G)$ can be set as previously described. As the problem is currently the only one, the global upper bound $\hat\wgt(G)$ is equal to the upper bound of $G$. Over the course of the algorithm, we repeatedly take a problem from $\queue$ and check whether we can reduce the graph size using our kernelization techniques outlined in Section~\ref{ss:kernel}. When possible, we perform the kernelization and push the kernelized problem to $\queue$. Otherwise, we branch on an edge $e$ adjacent to one of the terminals. 

The kernelization techniques detailed in Section~\ref{ss:kernel} reduce the size of the graph by finding edges that are (1) either guaranteed to be in a minimum multiterminal cut or (2) guaranteed not to be part of at least one minimum multiterminal cut. As we only want to find a single multiterminal cut with minimum sum of edge weights, we can delete edges in (1) and contract edges in (2). 

In Section~\ref{ss:branch} we detail the branching procedure which is used if these reduction techniques are unable to find any further reduction possibilities. For any edge~$e$, either it is in the multiterminal cut or it is not. We create two subproblems for $G$: $G/e$ and $G-e$. We aim to find the minimum multiterminal cut on either. Further details on the branching and edge selection are given in Section~\ref{ss:branch}. We compute upper and lower bounds for each of the problems and follow the branches whose lower bounds are lower than $\bestwgt$, the best cut weight previously found. In Section~\ref{ss:queue_impl} we discuss queue implementation and whether using a priority queue to first process 'promising' problems is useful in practice. We employ shared-memory parallelism by having multiple threads pull problems from $\queue$.

\section{Kernelization}
\label{ss:kernel}

We now show how to reduce the size of our graph to make the problem more manageable. This is achieved by contracting edges that are guaranteed not to be in the minimum multiterminal cut and deleting edges that are guaranteed to be in it. Before we detail the kernelization rules we show that edges not in $\cut$ can be safely contracted and edges in $\cut$ can be safely deleted if we store the weight sum of all deleted edges so far. The kernelization rules given in the following and outlined in Figure~\ref{fig:various_reductions} are used to identify such edges.

\begin{lemma}\label{lem:cont} \cite{cao20141}
  If an edge $e = (u,v) \in G$ is guaranteed not to be in at least one multiterminal cut $\cut(G)$ (\ie $P_u = P_v$), we can contract $e$ and $\wgt(G/e) = \wgt(G)$.
\end{lemma}

\begin{proof}
  As $e \not\in \cut(G)$, $\cut(G/e)$ is equal to $\cut(G)$ and thus still has weight equal to $w(\cut(G)) = \wgt(G)$. As an edge contraction only removes cuts and does not create any new cuts, an edge contraction can not lower the weight of the minimum multiterminal cut, \ie $\wgt(G/e) \geq \wgt(G)$. As $\cut(G/e)$ has weight $\wgt(G)$, it is a multiterminal cut in $G/e$ with weight equal to $\wgt(G)$. Thus it is definitely a minimum multiterminal cut with weight $\wgt(G)$.
\end{proof}

Lemma~\ref{lem:cont} allows us to reduce the graph size by contracting an edge if we can prove that both incident vertices are in the same partition in $\vopt$. The lemma can be generalized trivially to contract a connected vertex set by applying the lemma to each edge connecting two vertices of the set.

\begin{lemma}\label{lem:del} \cite{cao20141}
  If an edge $e = (u,v) \in E$ is guaranteed to be in a minimum multiterminal cut, \ie there is a minimum multiterminal cut $\cut(G)$ in which $P_u \neq P_v$, we can delete $e$ from $G$ and $\cut(G - e)$ is still a valid minimum multiterminal cut.
\end{lemma}

\begin{proof}
  Let $\wgt(G)$ be the weight of the minimum multiterminal cut $\cut(G)$. We show that for an edge $e \in \cut(G)$, $\wgt(G - e) = \wgt(G) - w(e)$. Thus, we can delete $e$ (and thus replace $G$ with $G - e$) and store the weight of the deleted edge. Obviously, $\cut(G - e)$ has weight equal to $\wgt(G) - w(e)$, as we just deleted $e$ and all other edges in $\cut(G)$ are still in $G$. By deleting $e$, the weight of any multiterminal cut can be decreased by at most $w(e)$ (as a multiterminal cut is a set of edges and $e$ can at most be once in that set). As $\wgt(G)$ is minimal by definition and no cut weight can be decreased by more than $w(e)$, $G - e$ cannot have a minimum multiterminal cut with weight $< \wgt(G) - w(e)$. Thus, $\cut(G-e)$ is a minimum multiterminal cut of $G-e$ with weight $\wgt(G-e)$.
\end{proof}

\paragraph{Minimum Isolating Cuts}
When we look at a problem, we first solve the minimum s-T-cut problem for each terminal $s \in T$. This results in one or multiple minimum cuts that separate $s$ from all other terminals. We call the side of the cut containing $s$ the \emph{isolating cut} of $s$. Dahlhaus~\etal\cite{dahlhaus1994complexity} prove that there is a minimum multiterminal cut $\cut$ in which the complete isolating cut is in $\vopt_s$. Thus, according to Lemma~\ref{lem:cont} we can contract all vertices of the largest isolating cut into a single vertex. In Figure~\ref{fig:example_mtcut} this would result in contracting the red areas into their respective terminals. This contraction might result in edges connecting terminals. Such an edge $e=(u,v)$, where both $u$ and $v$ are terminal vertices is guaranteed to be a part of $\cut(G)$. This comes from the fact that we know $\vopt_u \neq \vopt_v$, \ie $u$ and $v$ are not in the same block in the minimum multiterminal cut, as both $u$ and $v$ are terminals. According to Lemma~\ref{lem:del} they can therefore be deleted.

\subsection{Local Contraction}
We aim to find edges that cannot be part of the minimum multiterminal cut. If we find an edge that can be contracted, we mark it in a union find data structure~\cite{gabow1985linear}. This union-find structure is initialized with each vertex as its own block, an edge contraction then merges the two blocks of incident vertices. After all kernelization criteria are tested, we contract all edges that are marked as contractible. As a contraction might open up new contractions in its neighborhood, we run the contraction routines until they do not find any more contractible edges. To ensure low overhead, we run only the first iteration completely and subsequently check only the neighborhoods of vertices that were changed in the previous iteration.

\paragraph{Low-Degree Vertices~\cite{cao20141}}
Figures~\ref{fig:various_reductions}.(1),~\ref{fig:various_reductions}.(2) and \ref{fig:various_reductions}.(3) show examples of why non-terminal vertices with degree $\leq 2$ can be contracted. A non-terminal vertex with no neighbors (\texttt{IsolatedVertex}) can be deleted as there is no incident edge that could affect a cut. For a non-terminal vertex $v$ with only one adjacent edge $e = (v,x)$ (\texttt{DegreeOne}), $e$ can not be part of the minimum multiterminal cut $\cut(G)$. Any multiterminal cut that contains $e$ can be improved by removing $e$ and moving $v$ to the block of its neighbour $x$. Thus, we can contract $e$. On a non-terminal vertex with two adjacent edges $e_1$ and $e_2$ (\texttt{DegreeTwo}), the heavier edge $e_1$ can not be part of $\cut$, as replacing it with $e_2$ improves the cut value. If $e_1$ and $e_2$ have equal weight, we can contract either (but not both!). These reductions are performed in a single run, which we denote as \texttt{Low}.

\paragraph{Heavy Edges}
We now look to contract heavy edges. \texttt{HeavyEdge}~(\ref{fig:various_reductions}.(4)) and \texttt{HeavyTriangle}~(\ref{fig:various_reductions}.(5)) are reductions that were originally used for the minimum cut problem~\cite{Chekuri:1997:ESM:314161.314315,henzinger2018practical,padberg1990efficient}. We adapt them and transfer them to the minimum multiterminal cut problem. 

\texttt{HeavyEdge} says that an edge $e=(u,v)$ which has a weight of at least half of the total edge degree of a non-terminal vertex $u$ can be contracted, as any cut containing $e$ can instead also contain all other edges incident to $u$. If $e$ has at least $\frac{deg(u)}{2}$, all other incident edges together are not heavier.

For a \texttt{HeavyTriangle} with vertices $v_1$, $v_2$ and $v_3$, we can relax the condition. If for two of the vertices the incident triangle edges together are at least as heavy as all other incident edges, we can contract those, as shown in Figure~\ref{fig:various_reductions}.(5). Each of the continuous lines between $v_1$ and $v_2$ can be replaced with the dashed line without increasing the value of the cut. Thus, in every case ($v_3$ can be on either side of the cut), there is an optimal solution in which $v_1$ and $v_2$ are in the same block. Thus, we can contract the edge according to Lemma~\ref{lem:cont}.

\begin{figure*}[t]
  \centering
  \includegraphics[width=.9\textwidth]{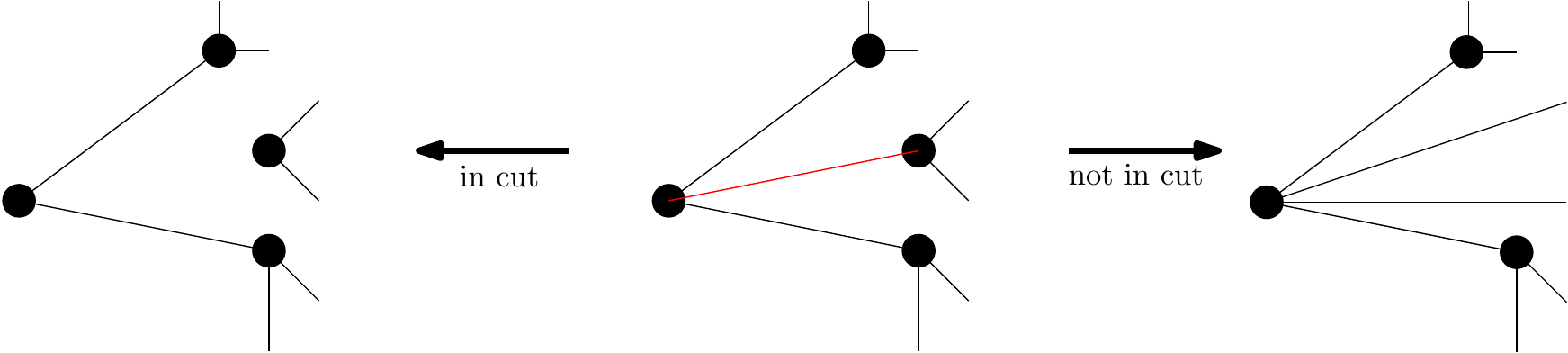}
  \caption{\label{fig:branch_edge}Branch on marked edge $e$ in $G$, adjacent to a terminal - create two subproblems, (1) $G/e$ and (2) $G-e$.}
\end{figure*}

The condition \texttt{SemiEnclosed}, shown in Figure~\ref{fig:various_reductions}.(6), considers a vertex $v$ which is mostly incident to terminal vertices. Let $t_1$ be the terminal that is most strongly connected to $v$ and $t_2$ the terminal with second highest connection strength. Now say that $v$ is contracted into any terminal vertex. All edges connecting $v$ with other terminals are then edges connecting terminals and are guaranteed to be in $\cut$. If $w(v, t_1) > w(v, t_2) + \sum_{u \in V \backslash T} w(v, u)$, \ie $(v, t_1)$ is heavier than the sum of $(v, t_2)$ and all edges connecting $v$ with non-terminals, we can contract $v$ into $t_1$. This follows from the fact that the weight of cut edges incident to $v$ is at most $deg(v) - w(v, t_1)$ if $v$ is in the same block as $t_1$. If we instead add $v$ to the block of $t_2$ (or any other block), at most $w(v, t_2) + \sum_{u \in V \backslash T} w(v, u)$ of the edges incident to $v$ would not be part of the cut. Thus, the locally best choice is contracting $v$ into $t_1$. As this does not affect any other graph areas, this choice is guaranteed to be optimal. We check both \texttt{HeavyEdge} and \texttt{SemiEnclosed} in a single run labelled \texttt{High}. \texttt{HeavyTriangle} is checked in a run named \texttt{Triangle}.

\subparagraph{High-connectivity edges}

The \emph{connectivity} of an edge $e=(u,v)$ is the value of the minimum cut separating $u$ and $v$. If an edge has connectivity $\geq \bestwgt(G)$, it is guaranteed that $u$ and $v$ are in the same block in $\vopt$, as there can not be a multiterminal cut that separates them and has value $<\bestwgt(G)$. We can therefore contract $u$ and $v$. We now show how to improve the bound.

\begin{lemma}\label{lem:noi}
  If for a graph $G$ with best known multiterminal cut $\bestcut(G)$, vertices $u$ and $v$ belong to different connected components of the minimum multiterminal cut $G\backslash \cut$, then $\lambda(u,v)+\frac{\sum_{i \in \{1,\dots,t\}\backslash \max_2} \lambda(G,t_i,T\backslash\{t_i\})}{4}  \leq |\wgt(G)|$, where $\max_2$ is the set of the indices of the largest $2$ values $\lambda(G,t_i,T\backslash\{t_i\})$ in the sum. 
\end{lemma}

\begin{proof}
In Appendix~\ref{app:proofs}
\end{proof}

We can use Lemma~\ref{lem:noi} to contract high-connectivity edges. This condition is denoted as \texttt{HighConnectivity}. For any edge $e=(u,v)$, if $\lambda(u,v)+\frac{\sum_{i \in \{1,\dots,k\}\backslash \max_2} \lambda(G,t_i,T\backslash\{t_i\})}{4} > |\wgt| \geq |\bestwgt|$, $u$ and $v$ are guaranteed to be in the same block in $\vopt$. Thus, we can contract them into a single vertex according to Lemma~\ref{lem:cont}.

As it is very expensive to compute the connectivity for every edge, we use the CAPFOREST algorithm of Nagamochi~\etal~\cite{henzinger2019shared,nagamochi1992computing,nagamochi1994implementing} to compute a connectivity lower bound $\gamma(u,v)$ for each edge $e = (u,v)$ in $G$ in near-linear time. If the lower bound $\gamma(u,v)$ fulfills Equation~\ref{eq:cap}, we can use Lemma~\ref{lem:noi} to contract $u$ and $v$.

\begin{equation}
  \label{eq:cap} \gamma(u,v) > |\hat\wgt| - \frac{\sum_{i \in \{1,\dots,k\}\backslash \max_2} \lambda(G,t_i,T\backslash\{t_i\})}{4}
\end{equation}

\section{Branching Tree Search}
\label{ss:branch}

If our reductions detailed in Section~\ref{ss:kernel} are unable to contract any edges in $G$, we branch on an edge adjacent to a terminal.  Figure~\ref{fig:branch_edge} shows an example in which we chose an edge to branch on. For each edge, there are two options: either the edge is part of the minimum multiterminal cut $\cut(G)$ or it is not. Lemmas~\ref{lem:cont}~and~\ref{lem:del} show that we can delete an edge that is in $\cut(G)$ and contract an edge that is not. Therefore we can build two subproblems, $G/e$ and $G-e$ and add them to the problem queue $\queue$.

Both of the subproblems will have a higher lower bound and thus, the algorithm will definitely terminate. For $G-e$, we know that $e$ is adjacent to a terminal $s$ but not an edge connecting two terminals (otherwise it would have been deleted). Thus, it is in exactly one minimum s-T-cut $\lambda(G,s,T\backslash\{s\})$. For the lower bound, we half the value of all minimum s-T-cuts. Deleting the edge indicates that it is definitely part of the multiterminal cut. Thus, we increased the lower bound by $w(e) - \frac{w(e)}{2} = \frac{w(e)}{2}$. 

For $G/e$ we know that $e=(s,v)$ is part of the largest isolating cut of $s$ (as we contract the largest isolating cut). In $G/e$ terminal $s$ is guaranteed to have a larger minimum s-T-cut, as otherwise there would be an isolating cut of equal value containing $v$, which contradicts the maximality of the contracted isolating cut. Thus $\lambda(G/e,s,T\backslash\{s\}) > \lambda(G,s,T\backslash\{s\})$ and no other minimum s-T-cut can be decreased by an edge contraction. Thus, the lower bound of $\wgt(G/e)$ and $\wgt(G-e)$ are both guaranteed to be higher than the lower bound of $\wgt(G)$.

\paragraph{Edge Selection}

In Section~\ref{ss:branch_impl} we evaluate the following edge selection strategies: \texttt{HeavyEdge} branches on the heaviest edge incident to a terminal; \texttt{HeavyVertex} branches on the edge between the heaviest vertex that is in the neighborhood of a terminal to that terminal; \texttt{Connection} searches the vertex that is most strongly connected to the set of terminals and branches on the heaviest edge connecting it to a terminal; \texttt{NonTerminalWeight} branches on the edge between the vertex that has the highest weight sum to non-terminal vertices and the terminal it is most strongly connected with; and \texttt{HeavyGlobal} branches on the heaviest edge in the~graph.  

\paragraph{Sub-problem order}

In Section~\ref{ss:queue_impl} we evaluate the following comparators for the priority queue $\queue$, \ie the order in which we look at the problems. A straightforward indicator on whether a problem can lead to a low cut is the current lower and upper bound for the best solution. If a problem has a good lower bound, it has a large potential for improvement and if it has a good upper bound there is already a good solution, potentially close to an even better solution in the neighborhood. Thus, \texttt{LowerBound} orders the problems by their lower bound and solves the ones with a better lower bound first while \texttt{UpperBound} first examines problems with a lower bound. In either comparator, the respective other bound acts as a tie breaker. \texttt{BoundSum} orders problems by the sum of their upper and lower bound.

\texttt{BiggerDistance} first examines problems in which the distance between lower and upper bound is very large. The conceptual idea is that those problems still have many unknowns and thus could be interesting to examine. In contrast to that, \texttt{LowerDistance} first examines problems with a lower distance of upper and lower bound, as those branches will likely have fewer subbranches. Following the same idea, \texttt{MostDeleted} first explores the problem that has the highest deleted weight. \texttt{SmallerGraph} orders the graphs by the number of vertices and first examines the smallest graph. As over the course of the algorithm a terminal might become isolated (as all incident edges were deleted), not all problems have the same amount of terminals. The isolated terminals are inactive and thus do not need any more flow computations. \texttt{FewTerminals} first examines problems with a lower number of active terminals. As there are many solutions with the same amount of terminals, ties are broken using \texttt{LowerBound}.

\section{Parallel Branch and Reduce}

Our algorithm is shared-memory parallel. As we maintain a queue of problems which are independent from each other, we can run our algorithm embarassingly parallel. The shared-memory priority queue of problems is implemented as a separate queue for each thread to pull from. When a thread adds a problem to the priority queue, it is added to a random queue with minimum queue size. In order to exploit data and cache locality, we add problems to the queue of the local thread if it is one of the queues with minimum size. Additionally, we fix each thread to a single CPU thread in order to actually use those locality benefits. In the beginning of the algorithm, there is only a single problem, which would leave all except for one processors idle, potentially for a long time, as we have to solve $k$ flow problems on the whole (potentially very large) graph. Thus, if there are idle processors, we distribute the flow problems over different threads.

\section{Combining Kernelization with ILP}

Multiterminal cut problems are generally solved in practice using integer linear programs~\cite{nabieva2005whole}. The following ILP formulation is adapted from \cite{henzinger2018ilp} and implemented using Gurobi 8.1.1. It is functionally equal to \cite{nabieva2005whole}.

  \begin{equation*}
\hspace*{-4cm}    \min \sum_{\{u,v\} \in E} e_{uv} \cdot w(\{u,v\})
  \end{equation*}

  \begin{equation*}
\hspace*{-2.5cm} \forall \{u,v\} \in E, \forall k: e_{uv} \geq x_{u,k} - x_{v,k}\\
  \end{equation*}
  \begin{equation*}
 \hspace*{-2.5cm}\forall \{u,v\} \in E, \forall k: e_{uv} \geq x_{v,k} - x_{u,k}\\
  \end{equation*}
  \begin{equation*}
 \hspace*{-4.5cm}   \forall v \in V: \sum_k x_{v,k} = 1\\
  \end{equation*}
  \begin{equation*}
  \hspace*{-4.65cm}  \forall i,j:  x_{t_i,j} = [ i = j ] \\
  \end{equation*}

\noindent Here, $x_{u,k}$ is $1$ iff vertex $u$ is in $V_k$ and $0$ otherwise and $e_{uv}$ is $1$ iff $(u,v)$ is a cut edge. 
We use this ILP formulation as a baseline of comparison. Additionally, we also create a new algorithm that combines the kernelization of our algorithm with integer linear programming. Using flow computations and kernelization routines, we are able to significantly reduce the size of most graphs while still preserving the minimum multiterminal cut. As the complexity of the ILP depends on the size of the graph and the complexity of the branch-and-reduce algorithm also depends on the value of the cut, this is fast on graphs with a high cut value in which the kernelization routines can reduce the graph to a very small size but with a large cut value. In the following, our algorithm \texttt{Kernel+ILP} first runs kernelization until no further reduction is possible and then solves the problem using the above integer linear programming formulation.

\begin{table}[hb!] \centering
  \small
   \caption{\label{t:ppigraphs}Large Real-world Benchmark Instances}
   \begin{tabular}{| l | r | r |}
     \hline
     Graph & $n$& $m$\\
     \hline \hline
     \multicolumn{3}{|l|}{Section~\ref{ss:social}: Social, Web and Map Graphs}\\
     \hline \hline
     bcsstk30 \cite{soper2004combined} & \numprint{28924} & $1.01M$\\
     ca-2010 \cite{bader2013graph} & $710K$ & $1.74M$\\
     ca-CondMat \cite{davis2011university} & \numprint{23133} & \numprint{93439}\\
     cit-HepPh \cite{davis2011university} & \numprint{34546} & $422K$\\
     eu-2005 \cite{BoVWFI} & $862K$ & $16.1M$\\
     higgs-twitter \cite{davis2011university} & $457K$ & $14.9M$\\
     in-2004 \cite{BoVWFI} & $1.38M$ & $13.6M$\\
     ny-2010 \cite{bader2013graph} & $350K$ & $855K$ \\
     uk-2002 \cite{BoVWFI} & $18.5M$ & $261M$ \\
     vibrobox \cite{soper2004combined} & \numprint{12328} & $165K$\\     
     \hline \hline
     \multicolumn{3}{|l|}{Section~\ref{ss:ppi}: Protein-protein Interaction}\\
     \hline \hline
     Acidithiobacillus ferrivorans & \numprint{3093} & \numprint{5394}\\
     Agaricus bisporus & \numprint{11271} & \numprint{14636} \\
     Candida maltosa & \numprint{5948} & \numprint{19462} \\
     Escherichia coli & \numprint{4127} & \numprint{13488} \\
     Erinaceus europaeus & \numprint{19578} & \numprint{68066} \\
     Homo sapiens & \numprint{19566}& $324K$ \\
     Mesoplasma florum & \numprint{683} & \numprint{2365} \\
     Saccharomyces cerevisiae & \numprint{6691} & \numprint{69809} \\
     Toxoplasma gondii & \numprint{7988} & \numprint{11779} \\
     Vitis vinifera & \numprint{29697} & \numprint{70206} \\
     \hline    
   \end{tabular}
 \end{table}

\section{Experiments and Results} \label{s:experiments}

We now perform an experimental evaluation of the proposed algorithms.
This is done in the following order: first analyze the impact of algorithmic components on our branch-and-reduce algorithm in a non-parallel setting, i.e.~we compare different variants for branching edge selection, priority queue comparator and the effects of the kernelization operators. We then report the speedup over ILP formulation and as well as parallel speedup on a variety of graphs. Lastly, we perform experiments on large real-world networks of various sources and protein-protein interaction graphs comparing \texttt{Kernel+ILP} and the branch-and-reduce algorithm.

\subsection{Experimental Setup and Methodology}

We implemented the algorithms using \CC-17 and compiled all
codes using g++-7.4.0 with full optimization (\texttt{-O3}). Our experiments are conducted on a
machine with two Intel Xeon Gold 6130 with 2.1GHz with 16 CPU cores each and
$256$ GB RAM in total. We perform five repetitions per instance and report average running~time. In this section we first describe experimental methodology. Afterwards, we evaluate different algorithmic choices in our algorithm and then we compare our algorithm to the state of the art. When we report a mean result we give the geometric mean as problems differ strongly in result and time.

\emph{Performance plots}
relate the fastest running time to the running time of each other algorithm on a per-instance basis.
For each algorithm, these ratios are sorted in increasing order. The plots show the ratio $t_\text{algorithm}/t_\text{best}$ on the y-axis. 
A point significantly above one indicates that the running time of the algorithm
was considerably worse than the fastest algorithm on the same instance. A value of one therefore indicates that the corresponding algorithm was one of the fastest algorithms to compute the solution.
Thus an algorithm is considered to outperform another algorithm if its corresponding values are below those of the other algorithm.

\subsubsection{Instances}

We use multiple set a instances to avoid overtuning the branch-and-reduce algorithm. To analyze the impact of algorithmic components in Sections~\ref{ss:branch_impl} to \ref{ss:kernelizationexp}, we  generate random hyperbolic graphs using the KaGen graph generator~\cite{funke2018communication}. These graphs have $n = 2^{14} - 2^{18}$ and an average degree of $8$, $16$ and $32$. For each graph size, we use three generated graphs and compute the multiterminal cut, each with $k \in \{3,4,5,6,7\}$. We use randomy hyperbolic graphs as they have power-law degree distribution and resemble a wide variety of real-world networks.
Additionally, we also use a family of weighted graphs from the $10^{th}$ DIMACS implementation challenge~\cite{bader2013graph}. These graphs depict US states, where a vertex depicts a census block and a weighted edge denotes the length of the border between two blocks. We use the $10$ states with fewest census blocks (AK, CT, DE, HI, ME, NH, NV, RI, SD, VT). For each state, we set the number of terminals $k \in \{3,4,5,6,7\}$. A multiterminal cut on these graphs depicts the shortest border that respects census blocks and separates a set of pre-defined blocks (or groups of blocks).
Here, we use one processor and set a timeout of $3$ minutes and a memory limit of $20$GiB. On these instances we also run the ILP and compare the results in Section~\ref{ss:comparisonilp} -- note that running the ILP itself without the kernelization rules on the large instances below is not feasible. We then look at the impact of parallelization for our algorithms in Section~\ref{ss:parallelbnr}.

When comparing \texttt{Kernel+ILP} with our branch and reduce framework in Sections~\ref{ss:social}~and~\ref{ss:ppi} on large instances, we use all $32$ cores of the machine (for the ILP as well as the branch and reduce framework). Here, we set a time limit of $1$ hour and a memory limit of $250$GiB and use the following graphs
In Section~\ref{ss:social} we perform experiments on $10$ large real-world networks of various sources. Table~\ref{t:ppigraphs} shows the sources and properties of the graphs. For each graph, we solve the minimum multiterminal cut problem for $k \in \{3,4,5,8\}$ terminals and $p \in \{10\%, 15\%, 20\%, 25\%\}$ vertices in the terminal. 
In Section~\ref{ss:ppi} we perform experiments on protein-protein interaction networks generated from the STRING protein interaction database~\cite{szklarczyk2010string,szklarczyk2018string} by using all edges they predict with a high certainty. We use the protein description to assign functions (block terminal affiliations) to proteins (vertices). We use the first occurence of a set of pre-defined function classes. For each graph, we examine problems with the $4,5,6,7,8$ most often occuring functions and with all (up to $15$, if all occuring in an organism) classes.

\begin{figure}[!t]
  \centering
  \begin{subfigure}{.5\textwidth}
    \includegraphics[width=\linewidth]{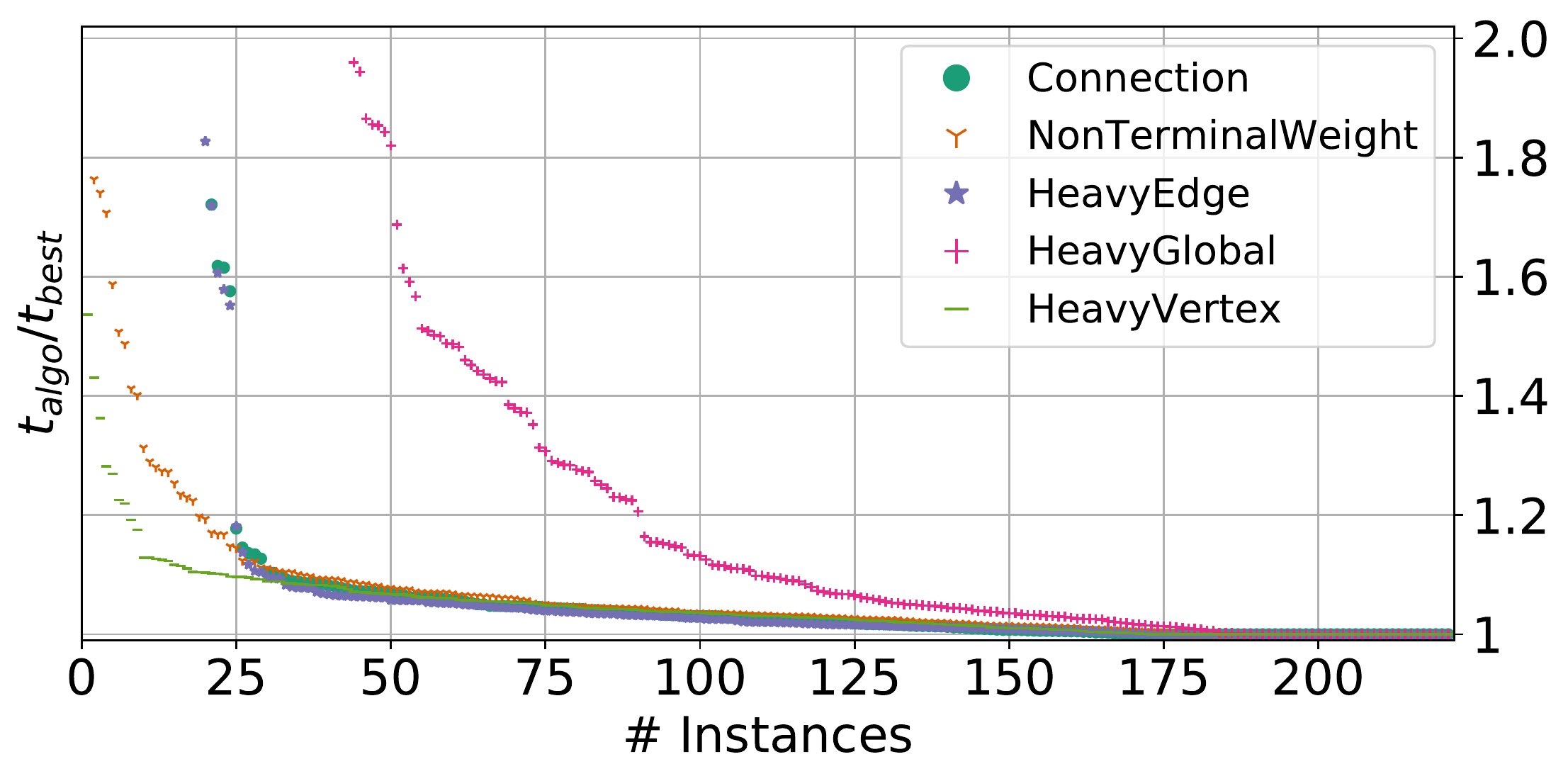}
    \caption{RHG graphs with partition centers as terminals}
    \label{fig:branch2}
  \end{subfigure}%
  \begin{subfigure}{.5\textwidth}
    \includegraphics[width=\linewidth]{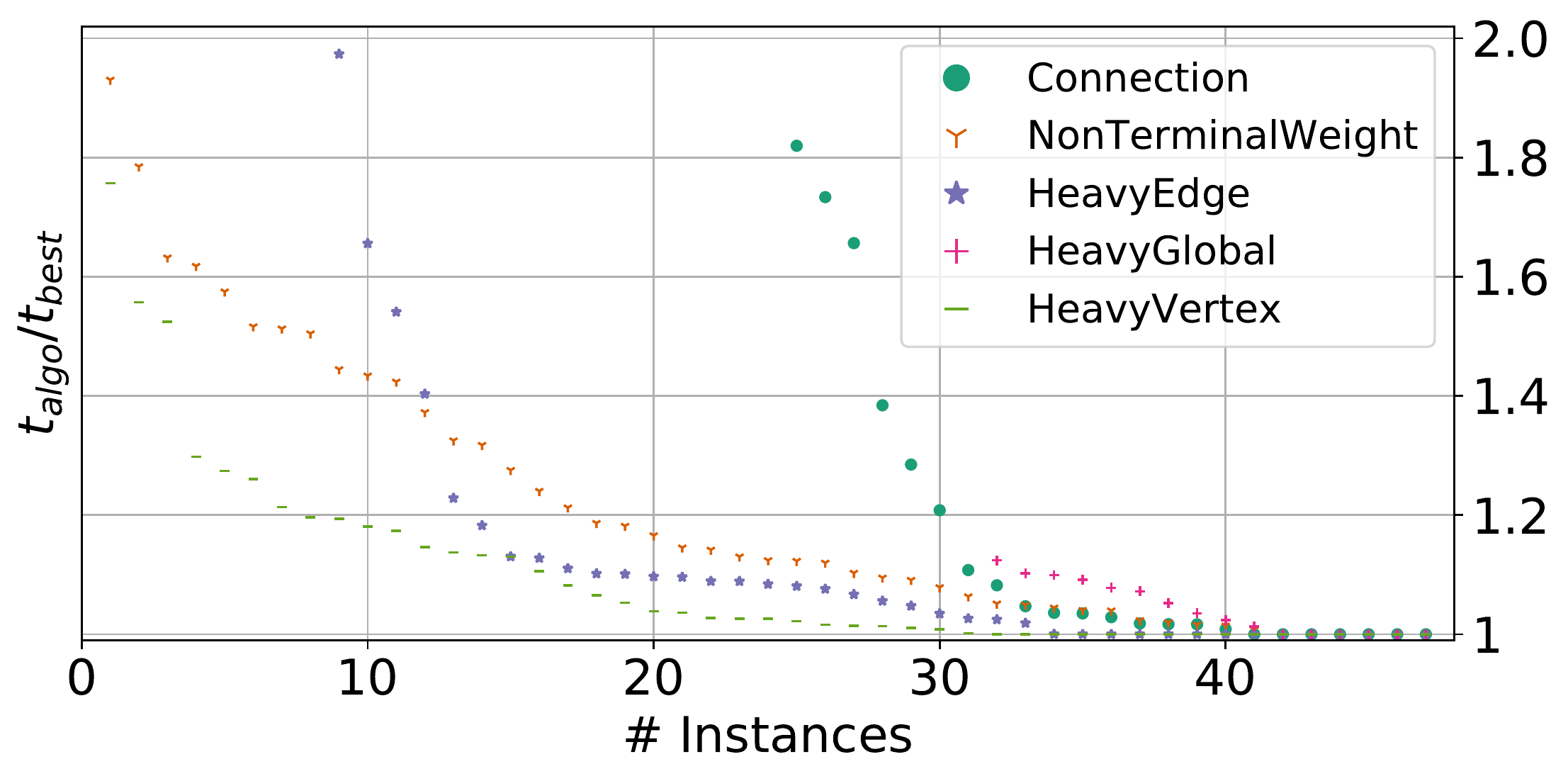}
    \caption{Map graphs with partition centers as terminals}
    \label{fig:branch4}
  \end{subfigure}
  \caption{Performance plots for branching edge selection variants}
  \label{fig:branch}
  \end{figure}

\subsection{Branching Edge Selection}\label{ss:branch_impl}

Figure~\ref{fig:branch} shows the results for the branching edge selection rules. In Subfigure~\ref{fig:branch2}, we show performance plots for RHG graphs and in Subfigure~\ref{fig:branch4} we show performance plots for map graphs. To find terminals, we partition the RHG graphs into $k$ parts and perform a breadth-first search starting in the block boundary. We define the vertex encountered last as the block center and use it as a terminal. In this experiment we use the \texttt{BoundSum} comparator and enable all kernelization rules.

As the minimum multiterminal cut of those problems usually turns out to be the trivial multiterminal cut of $k-1$ blocks of size $1$ and one block that comprises of the rest of the graph, we instead pick the last $10$ vertices encountered by the breadth-first search per block and contract them into a terminal. The minimum multiterminal cut of the resulting graph is usually not equal to the trivial multiterminal cut.

 In general, we aim to increase the lower bound by a large margin to reduce the number of subproblems that need to be checked. When we branch on a heavy edge, this increases the lower bound for $G-e$ by a large amount. For $G/e$, the lower bound is increased by half the amount of flow that is now added to the network. For a vertex that has a large number of edges to non-terminal vertices, contracting it into a terminal is expected to increase the flow by a large margin. The variant \texttt{HeavyVertex} chooses the edge $e$, for which the sum of edge weight and outgoing weights are maximized. It thus outperforms all other variants in both experiments. The only variant that is not guaranteed to be fixed-parameter tractable is \texttt{HeavyGlobal}, as this variant can also contract edges that are not incident to a terminal (and thus do not necessarily increase the lower bound). However, most edge contractions happen near terminals, so most heavy edges occur near terminals and thus \texttt{HeavyGlobal} often performs similar to \texttt{HeavyEdge}. 

 In all following experiments we use \texttt{HeavyVertex}, as it outperforms all other variants consistently.

 \begin{figure}[t]
  \centering
  \begin{subfigure}{.5\textwidth}
    \centering
    \includegraphics[width=\linewidth]{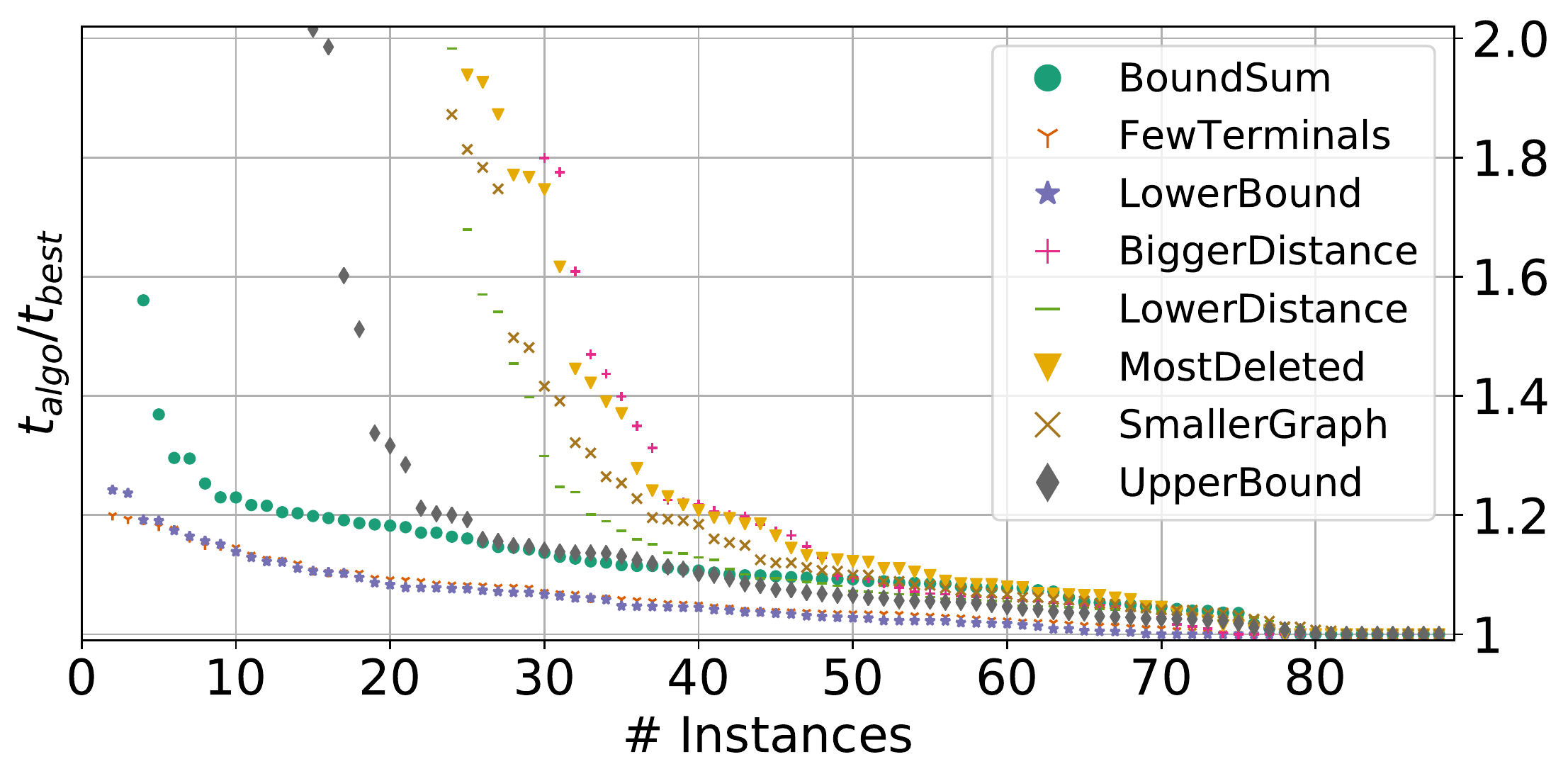}
    \caption{Graphs with $20\%$ of vertices in terminal}
    \label{fig:pq1}
  \end{subfigure}%
  \begin{subfigure}{.5\textwidth}
    \centering
    \includegraphics[width=\linewidth]{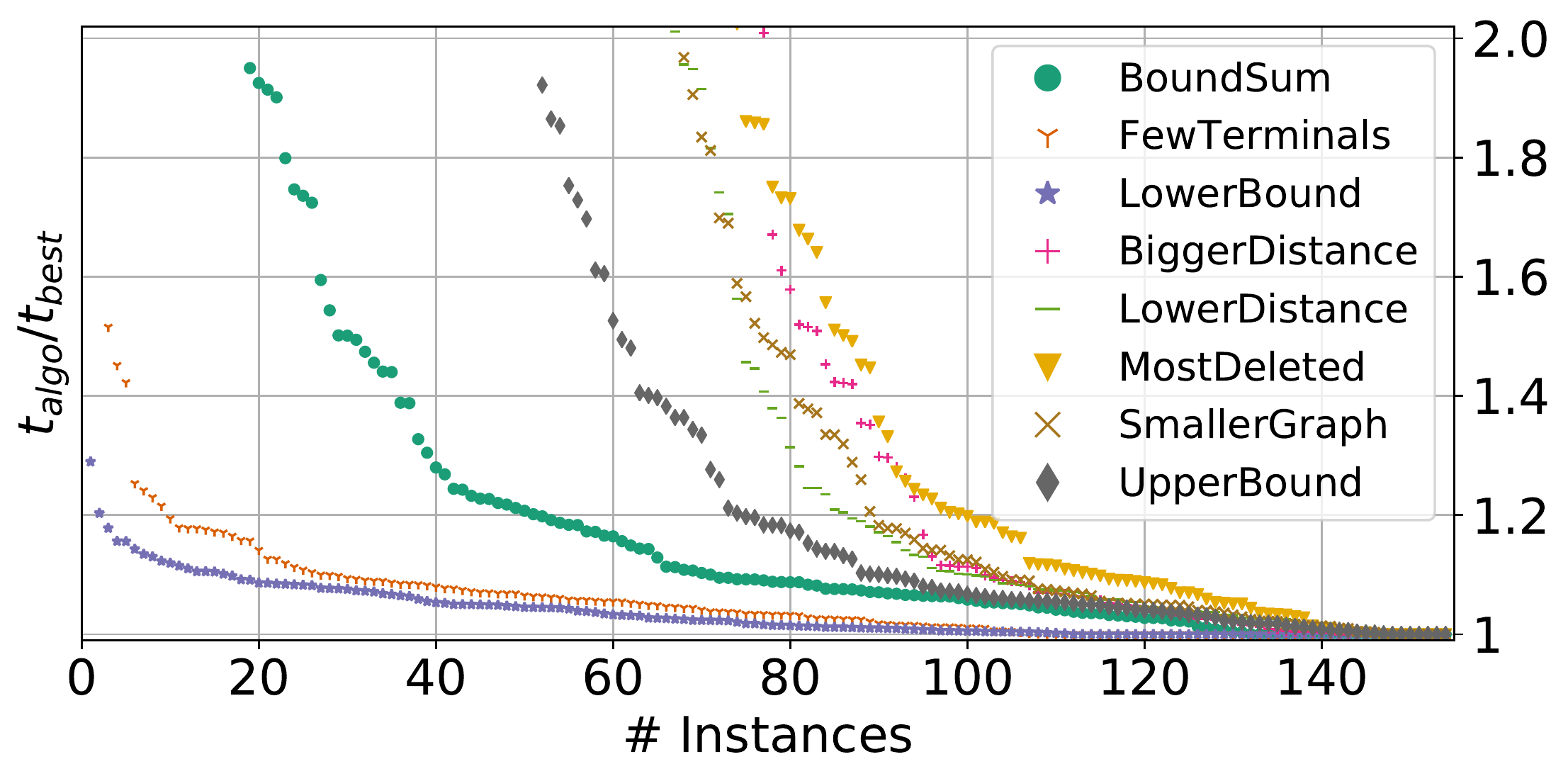}
    \caption{Graphs with $80\%$ of vertices in terminal}
    \label{fig:pq2}
  \end{subfigure}

  \caption{Performance plots for priority queue comparator variants}
  \label{fig:pq}
  \end{figure}

\subsection{Priority Queue Comparator}\label{ss:queue_impl}

We now explore the effect of the comparator used in the priority queue $\queue$. The choice of comparator decides which problems are highest priority and will be explored first. We want to first explore the problems and branches which will result in an improved solution, as this allows us to prune more branches. However, it is not obvious which criterion correctly identifies problems that might yield improved solutions, either directly on indirectly. Thus, we perform experiments on the same set of random hyperbolic and map graphs. 

On the random hyperbolic graphs examined in the previous experiment, the minimum multiterminal cut is often equal to the sum of all minimum-s-T-cuts excluding the heaviest. This is the cut that is found in the first iteration. If this is also the optimal cut, we definitely have to check all subproblems whose lower bound is lower than this cut. As the priority queue comparator only changes the order in which we examine those problems, the experimental results using the same problems as the previous section turned out very inconclusive. However, if we contract a sizable fraction of each block into its terminal, the minimum multiterminal cut is usually not equal to the union of s-T-cuts. Figure~\ref{fig:pq1} shows results for $20\%$ of vertices in the terminal on RHG graphs and Figure~\ref{fig:pq2} show results for $80\%$ of vertices in the terminal. 

\texttt{LowerBound} and \texttt{FewTerminals} are very competitive on most graphs. This indicates that problems with a low lower bound are very likely to yield improved results. The next fastest variant is \texttt{BoundSum}, which is almost competitive with $20\%$ of vertices in the terminal but significantly slower with $80\%$ of vertices in the terminal. However, \texttt{BoundSum} uses far less memory, as the lower bound of the newly created problems depends on the lower bound of the current problem. \texttt{BoundSum} examines many problems for which the lower bound is close to the currently best known solution. Thus, many newly created subproblems are immediately discarded when their lower bound is not lower than the currently best known solution. None of the other variants have noteworthy performance.

\begin{figure}[t!]
  \centering
  \begin{subfigure}{.5\textwidth}
   \centering
   \includegraphics[width=\linewidth]{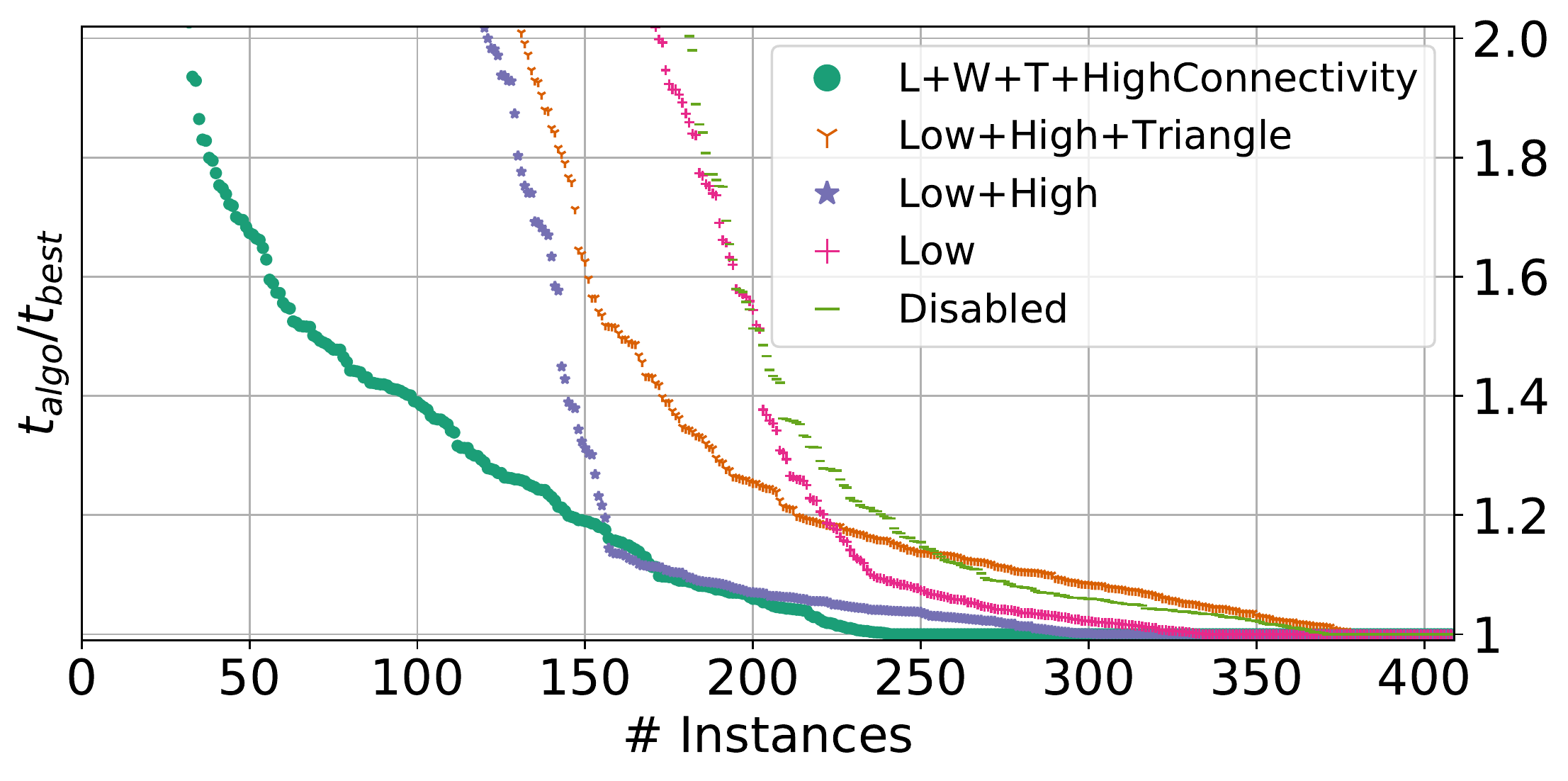}
   \caption{Impact of kernelization on RHG graphs.}
   \label{fig:kernel1}
 \end{subfigure}%
 \begin{subfigure}{.5\textwidth}
   \centering
   \includegraphics[width=\linewidth]{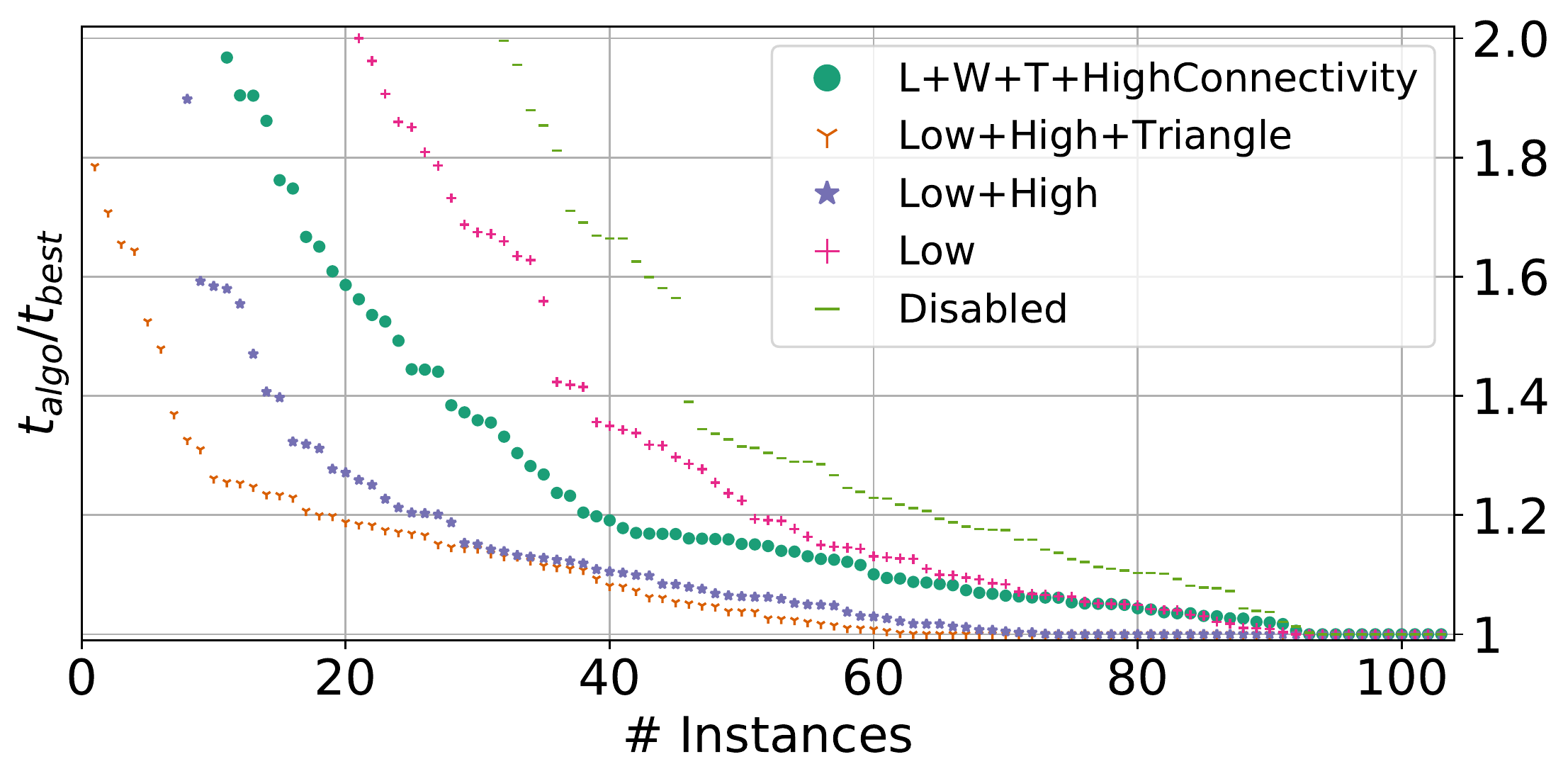}
   \caption{Impact of kernelization on map graphs.}
   \label{fig:kernel2}
 \end{subfigure}
 \caption{Performance plots for kernelization variants}
 \label{fig:kernel}
 \end{figure}

\subsection{Kernelization}
\label{ss:kernelizationexp}

We now study the effects of the kernelization operations performed in this work. For this purpose, we compare our algorithm without any kernelization to variants that enable different subsets of the kernelization operators detailed in Section~\ref{ss:kernel}. Figure~\ref{fig:kernel1} shows the results on the RHG graphs and Figure~\ref{fig:kernel2} shows the results on the map graphs. We use \texttt{BoundSum} as the priority queue comparator. In both cases, we combine results of $10$ vertices, $20\%$ and $80\%$ of vertices near block center in the terminal. The requirements in \texttt{Low} and \texttt{High} can be checked quickly whereas checking \texttt{Triangle} and \texttt{HighConnectivity} requires significant time. Thus, running \texttt{Low+High} is always useful, no matter how many edges can actually be contracted. On the RHG graphs in Figure~\ref{fig:kernel1}, \texttt{Triangle} does not find a lot of contractible edges that weren't already found by the previous kernelization operators. The high-degree vertices in the center of the hyperbolic plane have very high connectivity and thus, \texttt{HighConnectivity} is able to significantly reduce graph sizes and significantly improve running times compared to all other variants. In contrast, on the map graphs in Figure~\ref{fig:kernel2}, the connectivity of an edge can only be as high as the border length of the smaller vertex. Thus, we do not find many contractible edges. \texttt{Triangle}, however, is able to find many edges to contract, as vertices usually have few high degree neighbours. We can see that the utility of the kernelization operators depends heavily on the structure of the graph.

\begin{figure}
  \centering
  \begin{minipage}{.45\textwidth}
    \centering
    \raisebox{10pt}{
    \includegraphics[width=\linewidth]{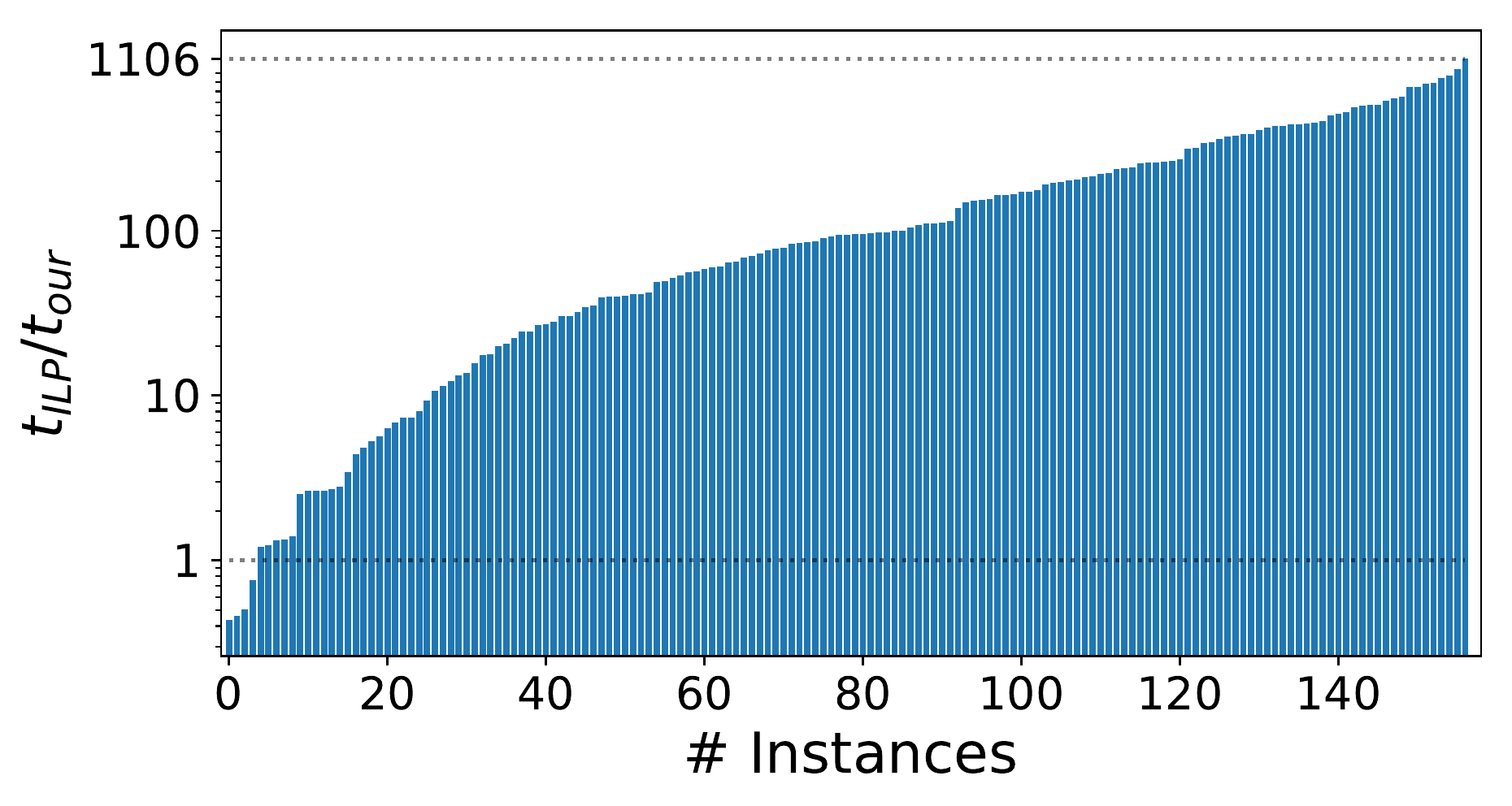}
    }
    \caption{Speedup of optimized branch-and-reduce to ILP}
    \label{fig:speeduptoilp}
  \end{minipage}%
  \hspace{5mm}
  \begin{minipage}{.45\textwidth}
    \centering
    \includegraphics[width=\textwidth]{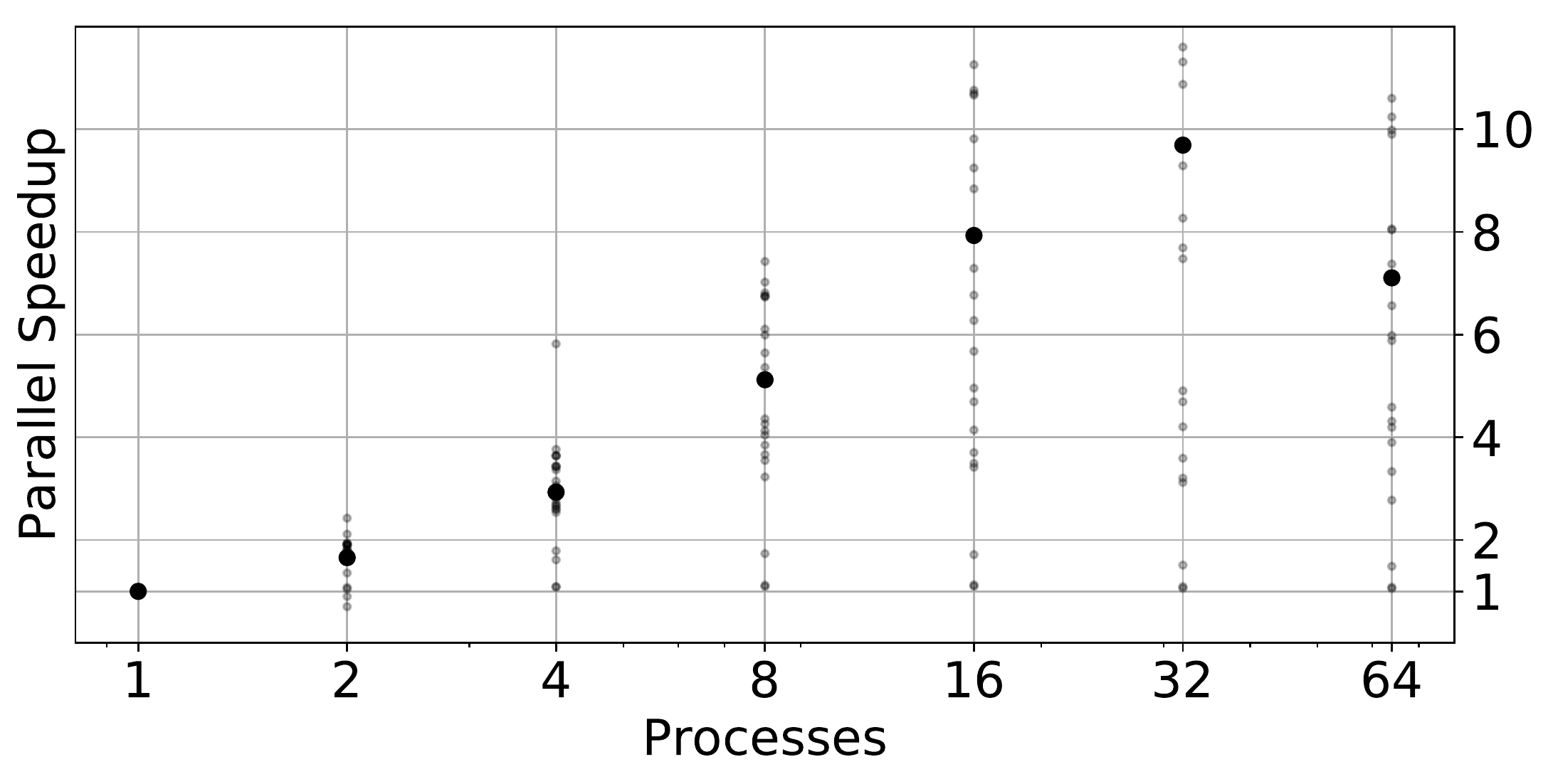}
    \caption{\label{fig:parallel}Parallel speedup on a variety of graphs. (low-alpha dot: one graph, solid dot: average speedup)}
  \end{minipage}
  \end{figure}

\subsection{Comparison to ILP}
\label{ss:comparisonilp}

Figure~\ref{fig:speeduptoilp} shows the speedup of the engineered branch-and-reduce algorithm, using \texttt{HeavyVertex} edge selection, \texttt{LowerBound} priority queue comparator and all kernelization rules enabled, to the ILP on all graphs from the previous subsections in which the ILP managed to find the minimum multiterminal cut within $3$ minutes. The branch-and-reduce algorithm outperforms the ILP on almost all graphs, often by multiple orders of magnitude. The ILP only solves $24\%$ of all problems, our algorithm solves $61\%$; on the problems solved by both, our optimized algorithm has a mean speedup factor of $67$, a median speedup factor of $95$ and a maximum speedup factor of $\numprint{1106}$. The mean speedup factor of the average of our algorithms compared to ILP is $43$ with a median speedup factor of $71$. Compared to the original ILP, \texttt{Kernel+ILP} is faster on all instances, has a mean speedup factor of $44$ and a median speedup factor of $49$. For figures, see Appendix~\ref{app:figures}.

This allows us to solve instances with more than a million vertices, while the ILP was unable to solve any instance with more than $\numprint{100000}$ vertices. As the basic ILP is unable to solve any large instances, we do not use it in the following experiments on large graphs.

\begin{figure*}[ht!]
  \centering
  \begin{subfigure}{.5\textwidth}
    \includegraphics[width=\linewidth]{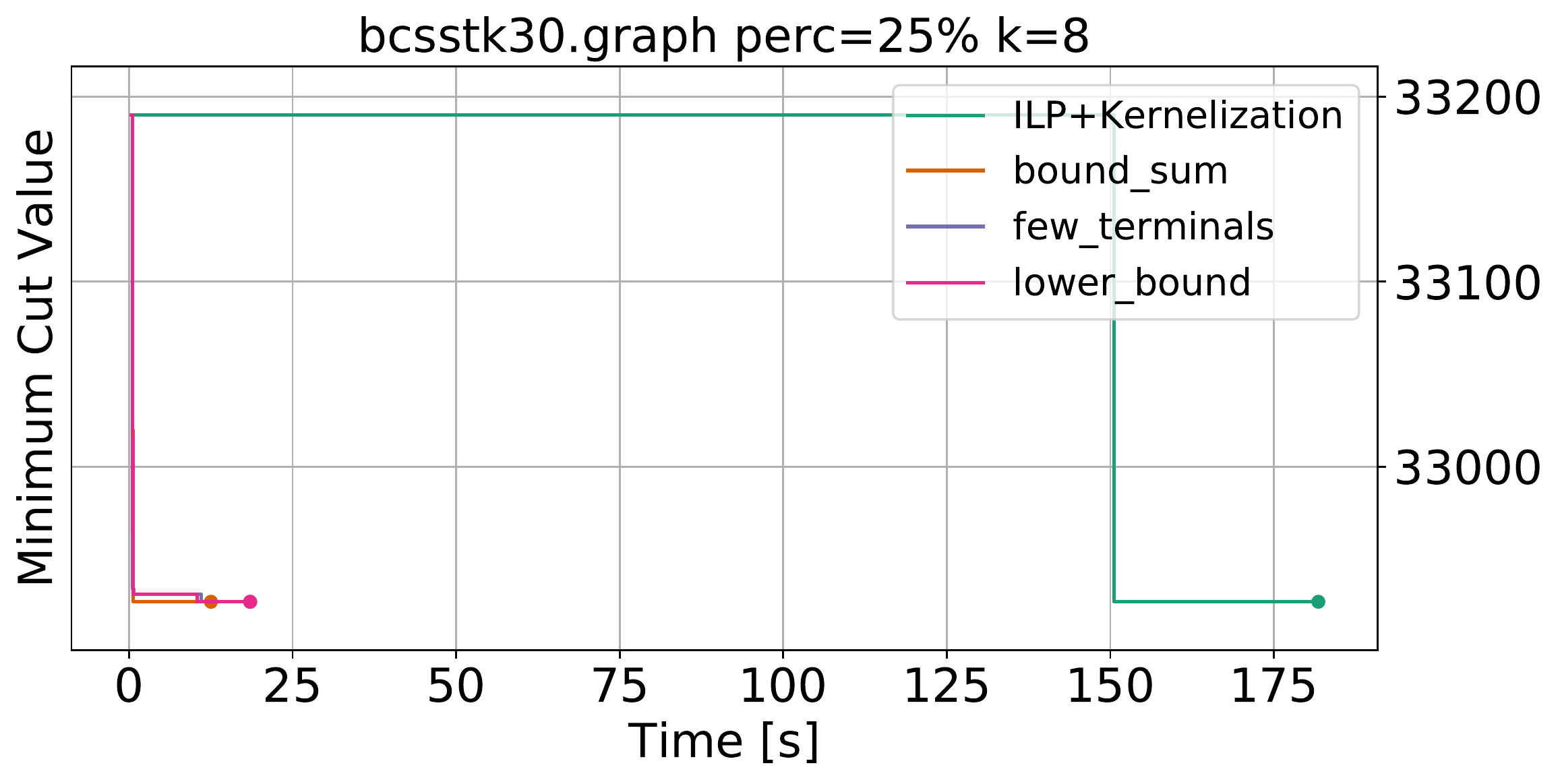}
  \end{subfigure}%
  \begin{subfigure}{.5\textwidth}
    \includegraphics[width=\linewidth]{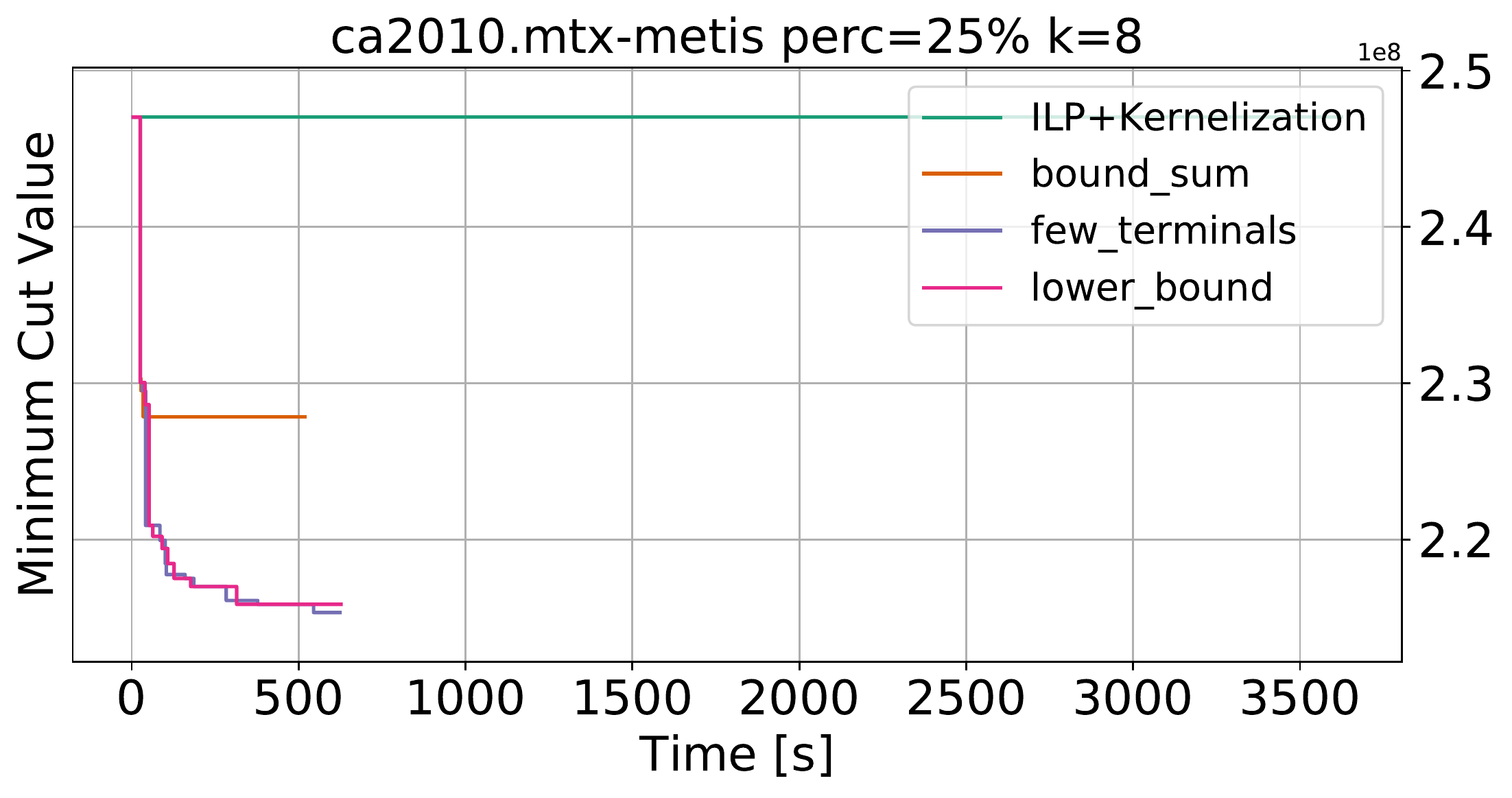}
  \end{subfigure}
  \begin{subfigure}{.5\textwidth}
    \includegraphics[width=\linewidth]{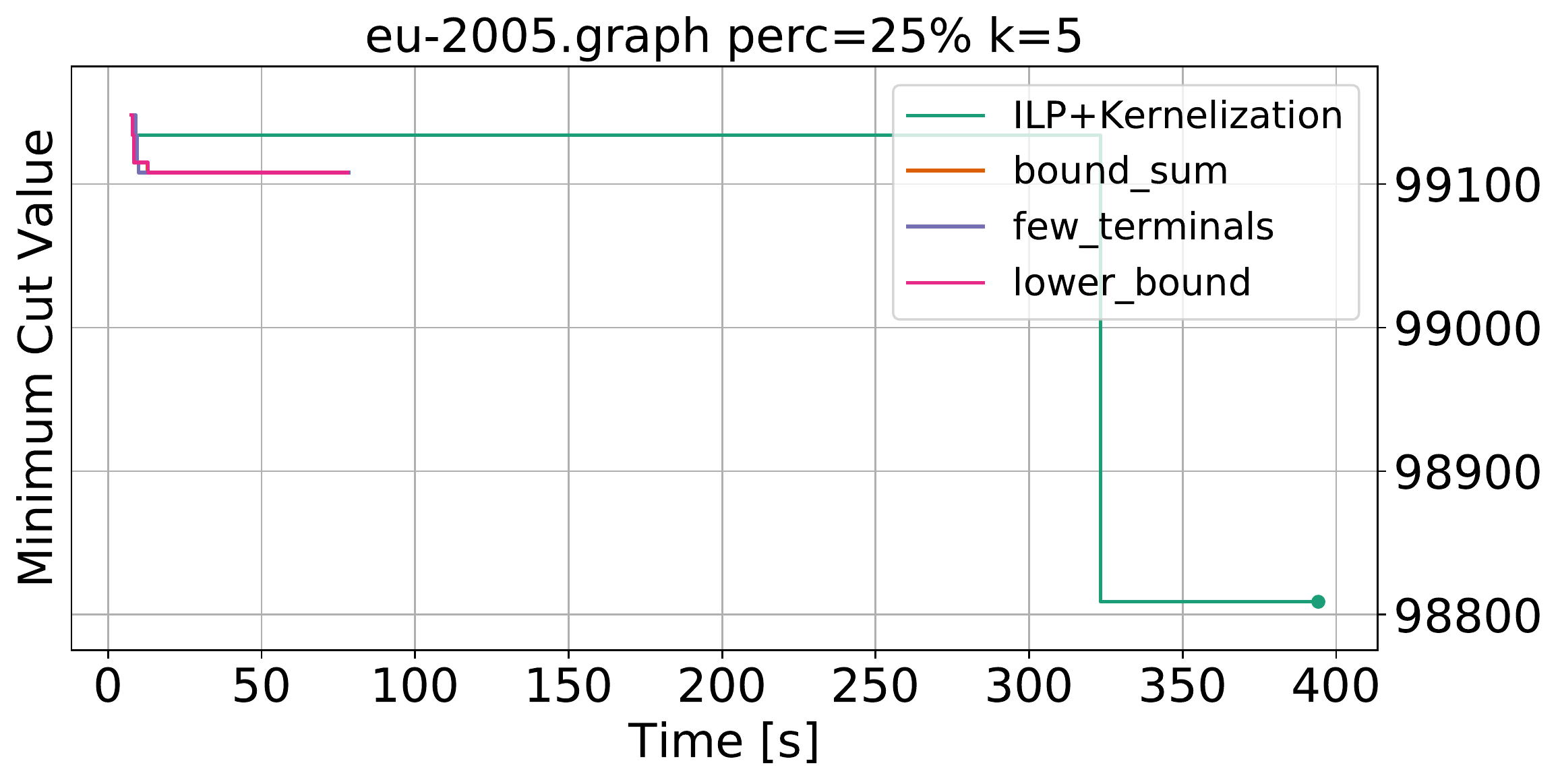}
  \end{subfigure}%
  \begin{subfigure}{.5\textwidth}
    \includegraphics[width=\linewidth]{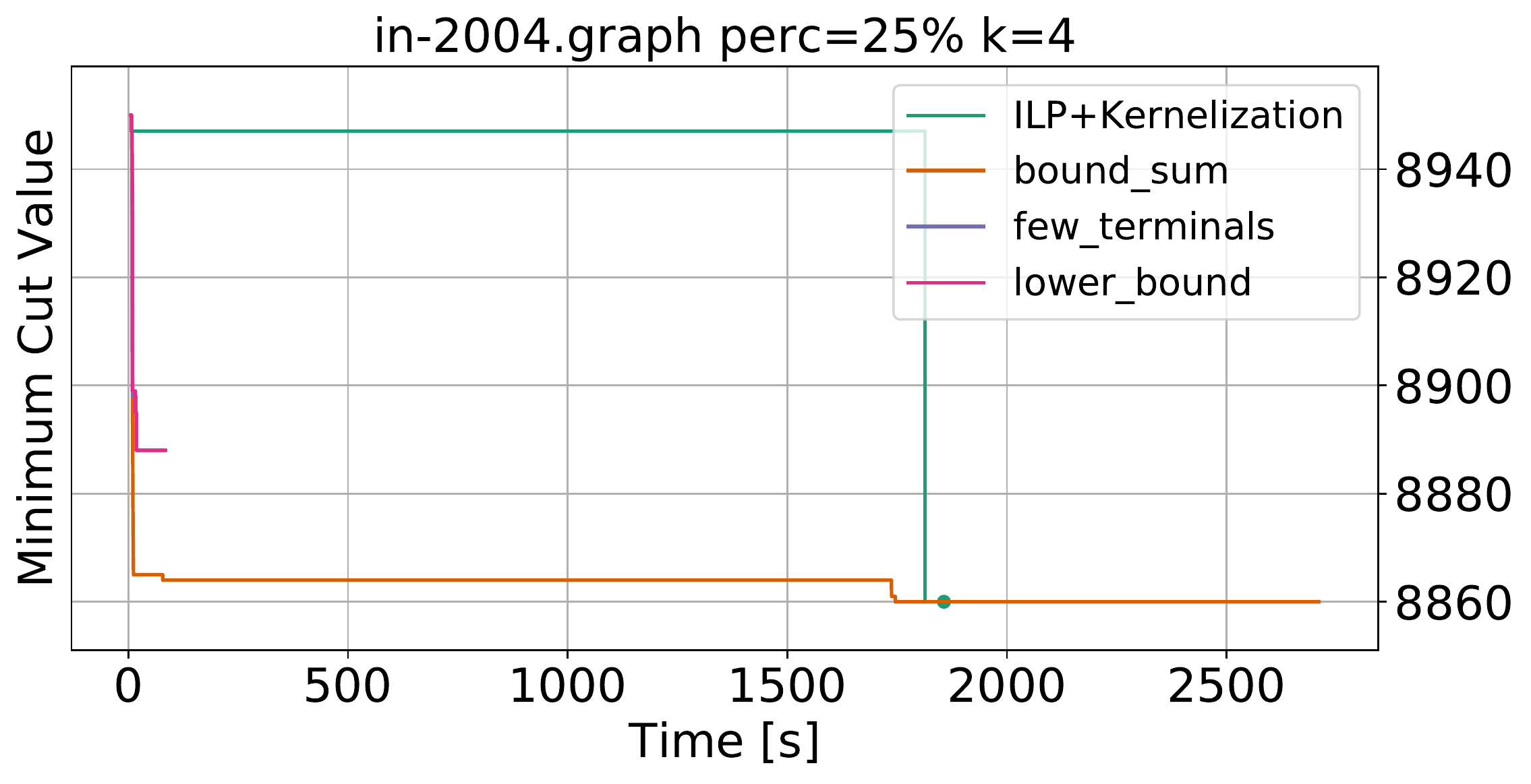}
  \end{subfigure}
  \begin{subfigure}{.5\textwidth}
    \includegraphics[width=\linewidth]{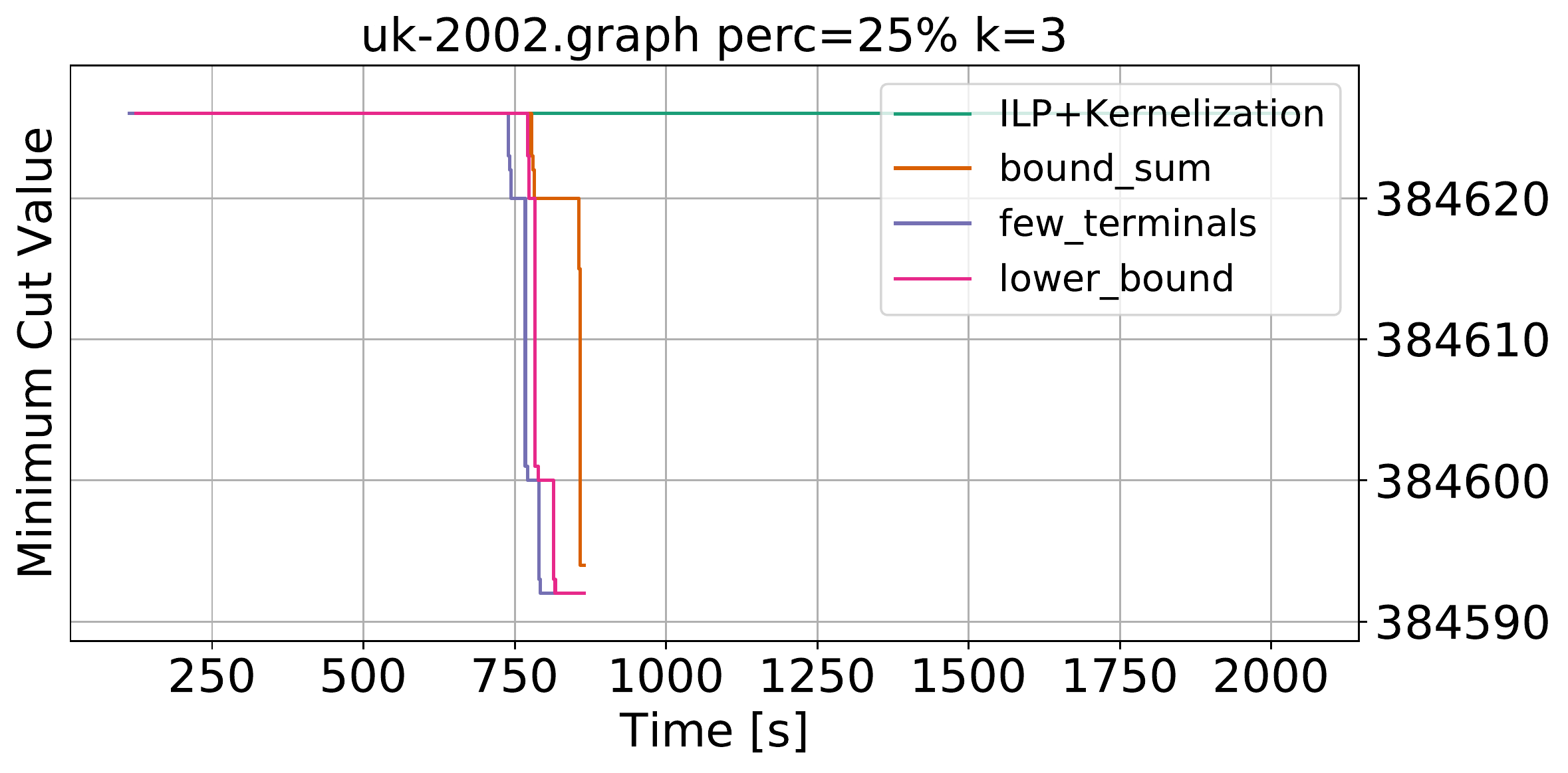}
  \end{subfigure}%
  \begin{subfigure}{.5\textwidth}
    \includegraphics[width=\linewidth]{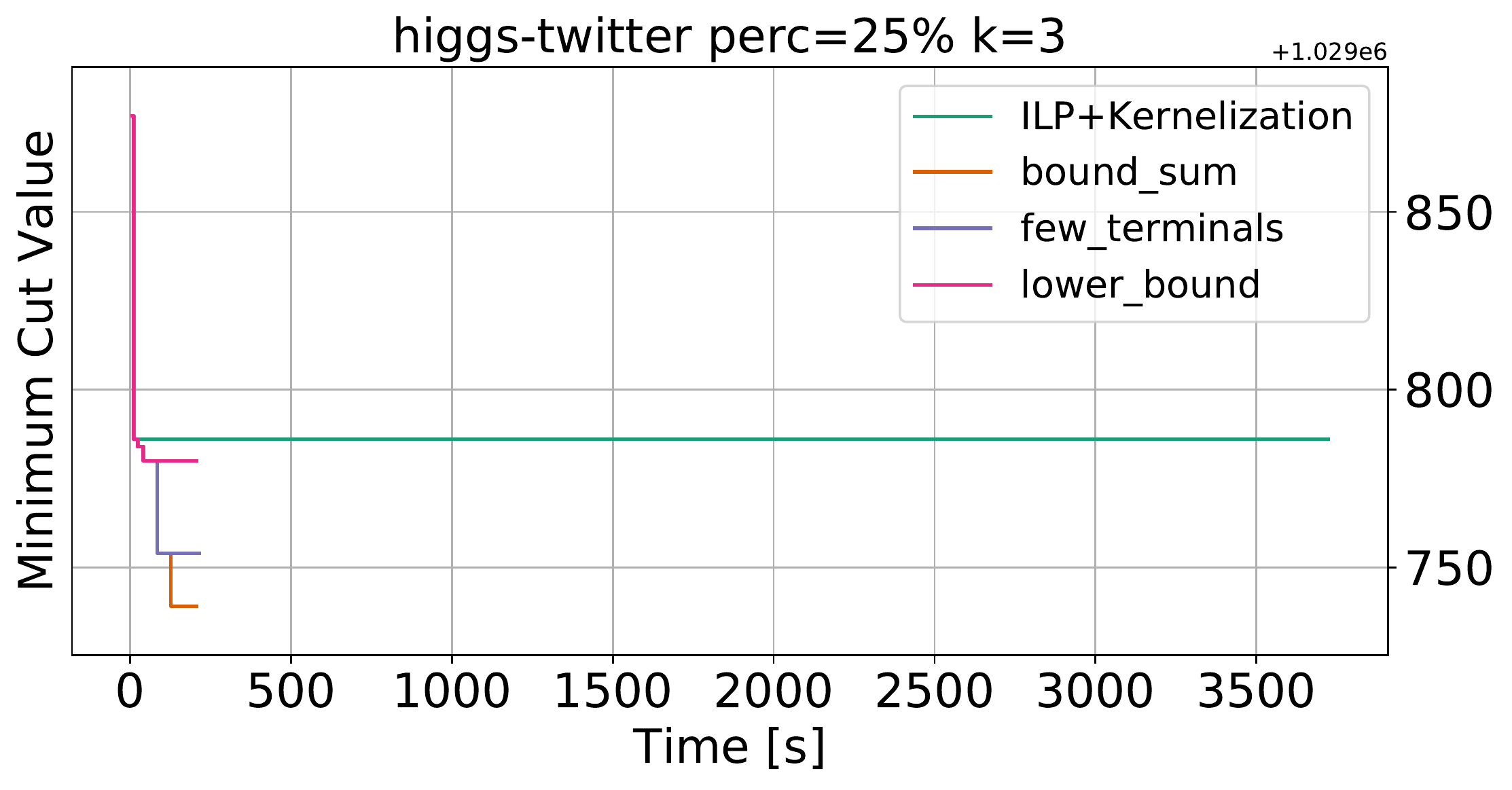}
  \end{subfigure}

  \caption{\label{fig:social} Progression of best result over time. Dot at end symbolizes that algorithm certifies optimality.}
\end{figure*}

\subsection{Parallel Branch and Reduce}
\label{ss:parallelbnr}

The previous experiments were all performed sequentially on a single thread. To see parallel speedup, we chose the $14$ RHG and $9$ map graphs that required the longest running time out of all terminated instances and solve them with varying amounts of processes. Figure~\ref{fig:parallel} shows the speedup for all graphs. The machine has $2$ processors with $16$ cores each and hyperthreading enabled. We can see that the average speedup factor is $7.9$x with $16$ threads, \ie one thread for each core of a single CPU. The speedup only slightly increases when using both CPUs and $32$ cores, to a factor of $9.7$x. This is caused by the large amount of data which needs to be transferred on the comparatively slower connection between the processors. When using hyperthreading, the average speedup even slows down to a factor of $7.1$x. 

On a few problems there is almost no parallel speedup.  These are the problems in which there is a large amount of work in a single graph and the flow problems are of starkly different difficulties. Thus, we would need to solve the flow problems in parallel to achieve a speedup in these graphs. As we only observed this behaviour in a few cases in the smaller RHG graphs and never in any large real-world networks, we refrained from including that additional complexity in our algorithm. If we exclude these $3$ problems from the problem set, we have an average speedup factor of $8.9$x with $16$ threads and of $10.9$x with $32$ threads. \texttt{Kernel+ILP} achieved almost no parallel speedup on these problems. This is the case as a large part of the solving time is spent in the sequential root relaxation.

\subsection{Large Real-World Networks}
\label{ss:social}

We aim to solve multiterminal cut problems on the large real-world networks in Table~\ref{t:ppigraphs} from a wide variety of graph and problem classes. Figure~\ref{fig:social} shows the progression of the best result over time for a set of interesting problems. Table~\ref{t:overview} gives an overview over the results. For each variant we show how often it produced the best result over all variants and how often it terminated with the optimal result. It also gives the mean result and time for all problems which were solved to optimality by all variants. Both in Figure~\ref{fig:social} and Table~\ref{t:overview} we can see that the branch and reduce variants find good solutions faster than \texttt{Kernel+ILP}. However, the variants often run out of memory in some of the largest instances. Especially in cases where the best multiterminal cut was already found (but not confirmed to be optimal) by the kernelization, \texttt{Kernel+ILP} managed to certify optimality more often than the branching variants. Thus it has the highest amount of terminated results, but reports significantly worse results in average. \texttt{Kernel+ILP} has about half as much improvements as the best variant \texttt{BoundSum}. In addition to giving the best results, variant \texttt{BoundSum} also has the lowest mean time in problems which were solved by all variants, however the improvement over the other branch-and-reduce variants is miniscule. The correlation between running time and number of vertices in the kernel graph is much stronger in \texttt{Kernel+ILP} compared to the branching variants.

\begin{table}[t] \centering
 \small
  \begin{tabular}{| l | r | r | r | r |}
    \hline
    Algorithm & \texttt{K+ILP} & \texttt{BSum} & \texttt{FTerm} & \texttt{LBound}\\
    \hline
    best result & \numprint{118} & \textbf{\numprint{136}} & \numprint{126} & \numprint{125}\\ 
    terminated & \textbf{\numprint{46}} & \numprint{35} & \numprint{33} & \numprint{33}\\
    mean result & \numprint{146570} & \textbf{\numprint{145961}} & \numprint{146052} & \numprint{146025}\\
    mean time & \numprint{18.69}s & \textbf{\numprint{6.71}s} & \numprint{6.97}s & \numprint{6.78}s\\
    \hline    
  \end{tabular}
  \caption{Overview of large real-world networks.\label{t:overview}}
\end{table}

\subsection{Protein-Protein Interaction Networks}
\label{ss:ppi}

Multiterminal cuts can be used for protein function prediction by creating a terminal for each possible protein function and adding all proteins which have this function to this terminal~\cite{karaoz2004whole,nabieva2005whole,vazquez2003global}. Table~\ref{t:ppiov} shows the results for these graphs. We can see that \texttt{Kernel+ILP} outperforms branch-and-reduce by a large margin on most graphs. This is the case because the kernelization is able to reduce the size of the graphs severely. These small problems with high cut values are better suited for \texttt{Kernel+ILP} than the branch-and-bound variants whose running time is more correlated with the value of the minimum multiterminal cut. The mean times are very low as some problems can be solved very quickly and thus drag the mean of all algorithms down. 

 \begin{table}[t] \centering
  \small
   \begin{tabular}{| l | r | r | r | r |}
     \hline
     Algorithm & \texttt{K+ILP} & \texttt{BSum} & \texttt{FTerm} & \texttt{LBound}\\
     \hline
     best result & \textbf{\numprint{57}} & \numprint{34} & \numprint{26} & \numprint{23}\\ 
     terminated & \textbf{\numprint{57}} & \numprint{25} & \numprint{23} & \numprint{21}\\
     mean result & \textbf{\numprint{4183}} & \numprint{4210} & \numprint{4218} & \numprint{4222}\\
     mean time & \textbf{\numprint{0.21}s} & \numprint{0.33}s & \numprint{0.36}s & \numprint{0.40}s\\
     \hline    
   \end{tabular}
   \caption{Overview of protein-protein-networks.\label{t:ppiov}}
 \end{table}

\section{Conclusion}\label{s:conclusion}

In this paper, we engineered data reduction rules for the minimum multiterminal cut problem with $k$ terminals. 
These reductions are used within a branch-and-reduce framework as well as to boost the performance of an ILP formulation for the problem.
Our experiments a) show that kernelization -- especially our newly introduced kernelization operators -- has a significant impact on both, the branch-and-reduce framework as well as the ILP formulation and b) 
show a clear trade-off: combining reduction rules with the ILP is very fast for problems which have a small kernel but a high cut value and the fixed-parameter tractable branch-and-reduce algorithm is highly efficient when the cut value is small.
Overall, we obtain algorithms that are multiple orders of magnitude faster than the ILP formulation which is de facto standard to solve the problem to optimality.
 Future work includes combining the branching algorithm with integer linear programming so that all occuring subproblems can be solved using the algorithm best suited for their problem properties. 
 We also aim to introduce scalable distributed parallelism.

\renewcommand{\bibname}{\begin{flushleft} References \end{flushleft}}
\bibliographystyle{abbrv}
\bibliography{paper}

\begin{appendix}

  \section{Proofs}
  \label{app:proofs}

In order to prove Lemma~\ref{lem:noi} we first prove the following useful claim:

\begin{claim}\label{t:blocks}
  For any two nodes $u$ and $v$, if $u$ and $v$ belong to different connected components of $G \backslash \cut(G)$, then $\lambda(u,v) \leq \frac{\sum_{i \in \{1,\dots,k\}}\delta (R(t_i))}{4} + \frac{\delta(R(u)) + \delta({R(v)})}{4}$, where $\delta$ are the weighted node degrees in the quotient graph corresponding to $\cut(G)$ and $R(x)$ is the block of a vertex $x$ as defined by the cut $\cut(G)$.
\end{claim}

\begin{proof}
  Let $G_R$ be the contracted graph where every block $R(t_i)$ in $G$ is contracted into a single vertex and let $|S(u,v)|$ be a minimum $u$-$v$-cut in $G_R$. 
  By definition of the minimum cut $\lambda(u,v)$, $\lambda(u,v) \leq |S(u,v)|$. 
  Slightly abusing the notation we denote by $R(t_i)$ the vertex of $G_R$ that results from contracting the block $R(t_i)$.
  
  For every vertex $w \in G_R$ that does not represent a block that contains either $u$ or $v$, at most $\frac{\text{deg}(w)}{2}$ edges are in $|S(u,v)|$. This follows directly from the assumption that $|S(u,v)|$ is minimal. If more than $\frac{\text{deg}(w)}{2}$ edges incident to $w$ are in $|S(u,v)|$, moving $w$ to the other side of the cut would give a better cut. Thus, at most half of the edges incident to $w$ are in $|S(u,v)|$. 
  
  We can not make this argument for the blocks containing $u$ and $v$, as potentially all edges incident to their blocks could be in the minimum multiterminal cut. Thus, $2 \cdot |S(u,v)| \leq \frac{\sum_{i \in \{1,\dots,k\}} \delta(R(t_i))}{2} + \frac{\delta(R(u))}{2} + \frac{\delta(R(v))}{2}$. The factor $2$ on the left side is caused by the fact that every edge is incident to two blocks. As we do not know the multiterminal cut $S$, we need to assume that they could be the blocks with the largest cuts $\delta(R(t_i))$. Dividing each side by $2$ finishes the proof.
\end{proof}

\begin{claim}\label{t:wgt}
  For any two nodes $u$ and $v$, if $u$ and $v$ belong to different connected components of $G \backslash \cut(G)$, then $\lambda(u,v) + \frac{\sum_{i \in \{1,\dots,k\}} \delta(R(t_i))}{4} \leq \wgt$.
\end{claim}
\begin{proof}
  Using Claim~\ref{t:blocks} we know that $\lambda(u,v) + \frac{\sum_{i \in \{1,\dots,k\}} \delta(R(t_i))}{4} \leq \frac{\sum_{i \in \{1,\dots,k\}} \delta(R(t_i))}{2}$. By definition of $\delta$, $\frac{\sum_{i \in \{1,\dots,k\}} \delta(R(t_i))}{2} = \wgt(G)$.
\end{proof}

We now use Claims~\ref{t:blocks}~and~\ref{t:wgt} to prove Lemma~\ref{lem:noi}.
\begin{proof}
  Let vertices $u$ and $v$ be in different blocks. Then     
  $\lambda(u,v)+\frac{\sum_{i \in \{1,\dots,t\}\backslash \max_2} \lambda(G,t_i,T\backslash\{t_i\})}{4} \leq$\\ 
  $\lambda(u,v)+\frac{\sum_{i \in \{1,\dots,t\}\backslash \max_2} \delta(R(t_i))}{4} \leq \frac{\sum_{i \in \{1,\dots,t\}\backslash \max_2} \delta(R(t_i))}{2} = \wgt(G)$.

  The first inequality follows from the fact that $\lambda$ is per definition the minimal cut separating $t$ from $T\backslash\{t_i\}$ and thus $\lambda(G,t_i,T\backslash\{t_i\}) \leq \delta(R(t_i))$.

  Thus, we know that if $\lambda(u,v)+\frac{\sum_{i \in \{1,\dots,t\}\backslash \max_2} \lambda(G,t_i,T\backslash\{t_i\})}{4} > \wgt(G)$, $u$ and $v$ are in the same block and the edge connecting them can be safely contracted.
\end{proof}
 
    \section{Additional Figures}
    \label{app:figures}
    \begin{figure}[!h]
      \centering
      \includegraphics[width=.5\linewidth]{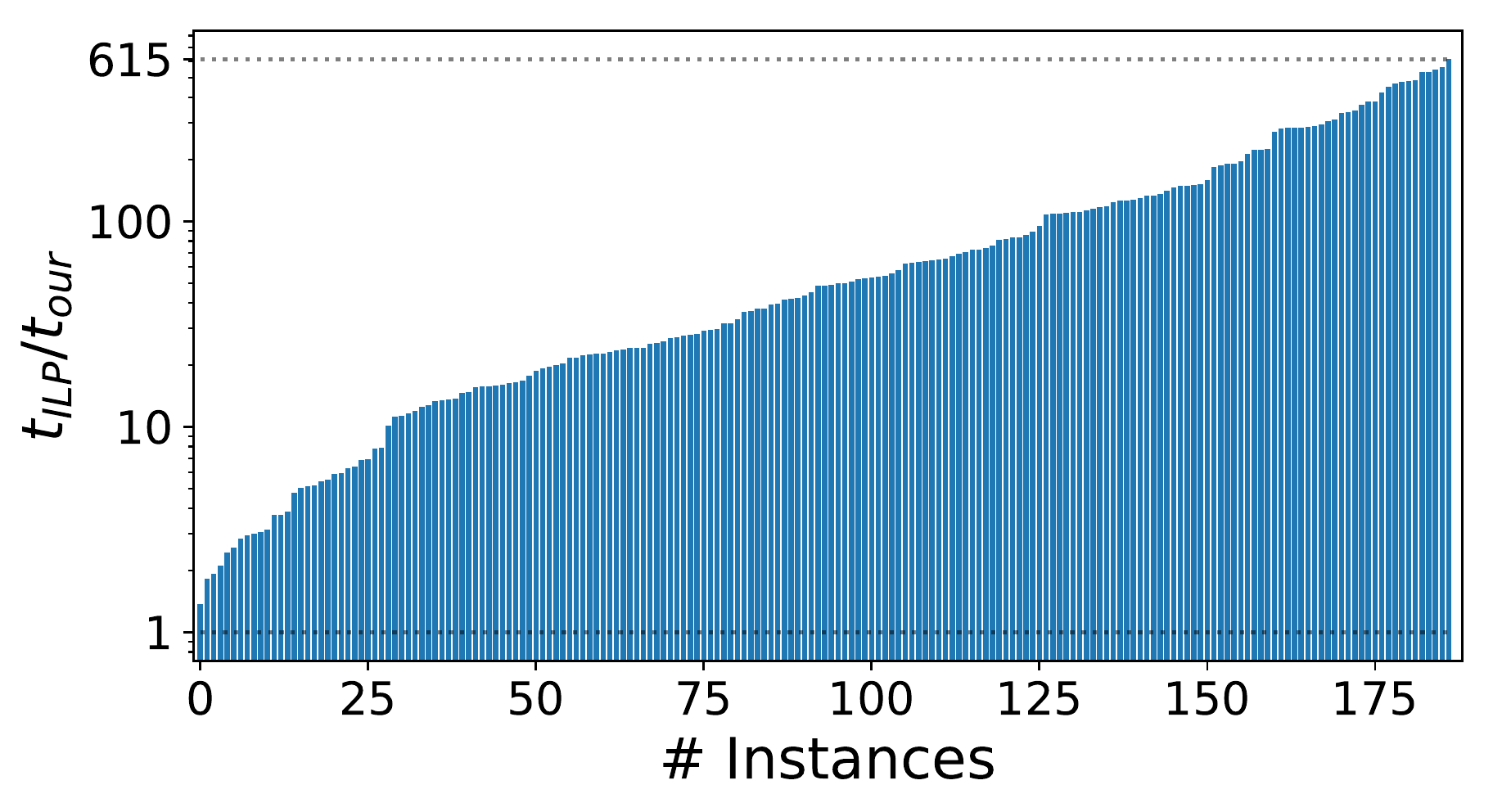}
      \caption{Speedup of \texttt{Kernel+ILP} to ILP}
      \label{fig:nbtoilp}
    \end{figure}
  
    \begin{figure}[!ht]
      \centering
      \includegraphics[width=.5\linewidth]{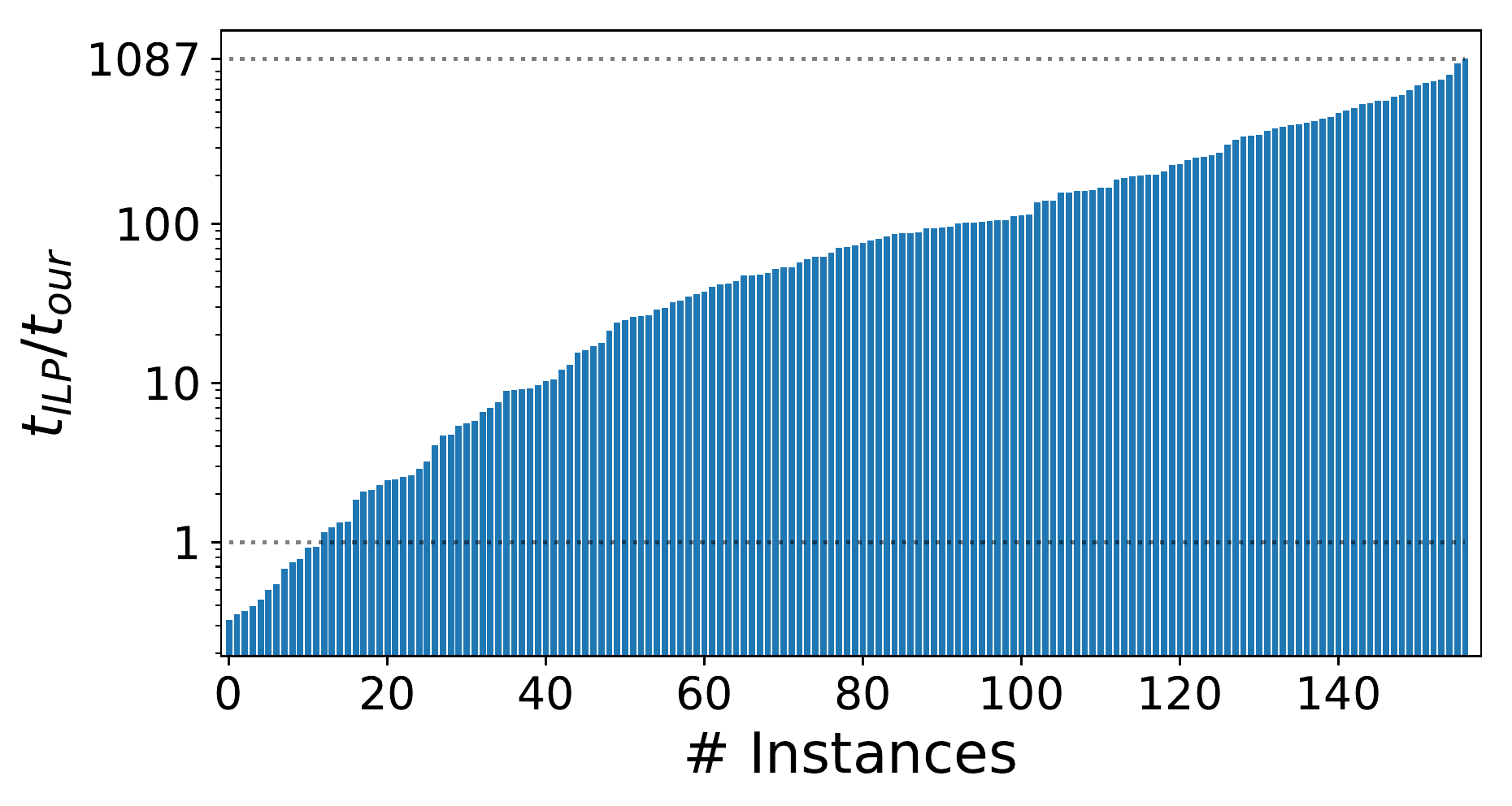}
      \caption{Speedup of avg. branch-and-reduce to ILP}
      \label{fig:alltoilp}
    \end{figure}

\end{appendix}

\end{document}